\documentclass{amsart}%
\usepackage{amsfonts}
\usepackage{amsmath}
\usepackage{amssymb}
\usepackage{graphicx}%
\newtheorem{theorem}{Theorem}
\theoremstyle{plain}

\newtheorem{definition}{Definition}

\newtheorem{lemma}{Lemma}
\newtheorem{notation}{Notation}

\newtheorem{proposition}{Proposition}
\newtheorem{remark}{Remark}

\numberwithin{equation}{section}
\begin{document}
\title[Non-Archimedian Replicator Dynamics]{ Eigen's Paradox \ and the quasispecies model in a \ Non-Archimedean Framework}
\author[Z\'{u}\~{n}iga-Galindo]{W. A. Z\'{u}\~{n}iga-Galindo}
\address{University of Texas Rio Grande Valley\\
School of Mathematical \& Statistical Sciences\\
One West University Blvd\\
Brownsville, TX 78520, United States }
\email{wilson.zunigagalindo@utrgv.edu}
\thanks{The author was partially supported by the Lokenath Debnath Endowed Professorship.}

\begin{abstract}
In this article we present a new $p$-adic generalization of the Eigen-Schuster
model where the genomes (sequences) are represented by words written in the
alphabet $\left\{  0,1,\ldots,p-1\right\}  $, where $p$ is a prime number,
with a time variable length. The time evolution of the concentration of a
given sequence is controlled by a $p$-adic evolution equation. The long term
behavior of the concentration of a sequence depends on a fitness function $f$,
a mutation measure $Q$, and an initial concentration distribution. The new
model provides essentially two types of asymptotic scenarios for evolution. If
the complexity of sequences grows at the right pace, then in the long term the
survival is assured. This agrees with the fact that larger genome size
improves the replication fidelity. In other case, the sequences cannot copy
themselves with sufficiently fidelity, and in the long term they will not
survive. Eigen's paradox is one, among the infinitely many, possible scenarios
of the evolution in the long term. The mathematical formulation of this fact
requires solving the Cauchy problem for the $p$-adic Eigen-Schuster model in a
rigorous mathematical way. This requires imposing restrictions on the fitness
function and on the mutation measure, among other conditions. The study of the
mentioned initial value problem requires techniques of $p$-adic wavelets and
$p$-adic heat kernels developed in the last 35 years.

\end{abstract}
\maketitle

\section{Introduction}

A central problem in the origin of life is the reproduction of primitive
organisms with sufficient fidelity to maintain the information coded in the
primitive genomes. Assuming\ that genomes have constant length,\ and the
existence of independent point mutations, that is, assuming that during the
replication process each nucleotide has a fixed probability of being replaced
for another nucleotide, and that this probability is independent \ of all
other nucleotides, Eigen discovered that the mutation process places a limit
on the number of nucleotides that a genome may have, see e.g. \cite{Eigen1971}%
, \cite{Eigen et al}, \cite{Nowak}, \cite{SchusterPeter}, \cite{Tannembaum et
al}. This critical size is called the error threshold of replication. The
genomes larger than this error threshold will be unable to copy themselves
with sufficiently fidelity, and the mutation process will destroy the
information in subsequent generations of these genomes. This contradicts the
existence of large stable living organisms on earth. To create more complex
organisms (that is to have more genetic complexity), it is necessary to encode
more information in larger genomes by using a replication mechanism with
greater fidelity. But the information for creating error-correcting mechanisms
(enzymes) should be encoded in the genomes, which have a limited size. Hence,
we arrive to the `Catch-22' or Eigen's paradox of the origin of life:
\textquotedblleft no large genome without enzymes, and no enzymes without a
large genome,\textquotedblright\ see \cite[p. 317]{Maynad Smith},
\cite{Szat-tree}.

In \cite{Sch} Scheuring and in \cite{Poole etal}, Poole, Jeffares, Penny
proposed biological hypotheses to explain Eigen's paradox. In both works, the
authors pointed out that escaping the Catch-22 (Eigen's paradox) requires that
the length of the genomes must growth at the right pace. The standard
Eigen-Schuster model is not compatible with the assumption that the length of
the sequences\ growth in time. First, the description of the space of
sequences as a metric space of binary sequences endowed with the Hamming
distance becomes useless. The Hamming distance make sense only when the length
of the sequences is finite and fixed. The classical realization of the
Eigen-Schuster model as a system of ODEs in $\mathbb{R}^{n}$ is useless,
because $n$ is the number of sequences (chemical species), if the length of
the sequences growth in time, then the number of chemical species grows and
consequently $n$ must grow in time. In conclusion, dealing with the assumption
that the length of the sequences growth in time requires a new mathematical approach.

In \cite{Zuniga-JPhA}, the author introduced a new non-Archimedean model of
evolutionary dyna\-mics, where the genomes (sequences) are represented by
$p$-adic numbers. The length of the sequences varies in time, and it is not
bounded. The sequences are organized in a tree-like structures resembling the
phylogenetic trees. There is a natural distance between two sequences which
depends on the first common ancestor of the given sequences. The space of all
possible sequences has a fractal nature. The time evolution of the
concentration of a sequence is controlled by a $p$-adic evolution equation,
which is a $p$-adic continuous version of the classical Eigen-Schuster model.
This equation depends on a fitness function $f$ and on mutation measure $Q$.
For some families of mutation measures and by using a $p$-adic version of the
Maynard Smith Ansatz, in \cite{Zuniga-JPhA}, the author showed the existence
of threshold function $M_{c}(f,Q)$, such that the long term survival of a
sequence requires that its length grows faster than $M_{c}(f,Q)$. This implies
that Eigen's paradox does not occur if the complexity of genomes grows at the
right pace. In \cite{Zuniga-JPhA} a heuristic approach to the existence of
quasispecies was presented, the purpose of this work is to provide a rigorous
mathematical analysis of this model, that allows explaining the Eigen paradox.
This requires solving, in a rigorous mathematical way, the Cauchy problem
attached to our $p$-adic Eigen-Schuster model under general assumptions on the
fitness function, the mutation measure, and the initial datum. Also, it is
needed to explain how the old Eigen-Schuster model fits in the new framework.

In the non-Archimedean model a sequence (genome) is specified by a $p$-adic
number:%
\begin{equation}
x=x_{-k}p^{-k}+x_{-k+1}p^{-k+1}+\ldots+x_{0}+x_{1}p+\ldots,\text{ with }%
x_{-k}\neq0\text{,} \label{p-adic-number}%
\end{equation}
where $p$ denotes a fixed prime number, and the $x_{j}$s \ are $p$-adic
digits, i.e. numbers in the set $\left\{  0,1,\ldots,p-1\right\}  $. The set
of all possible sequences constitutes the field of $p$-adic numbers
$\mathbb{Q}_{p}$. There are natural field operations, sum and multiplication,
on series of form (\ref{p-adic-number}), see e.g. \cite{Koblitz}. There is
also a natural norm in $\mathbb{Q}_{p}$ defined as $\left\vert x\right\vert
_{p}=p^{k}$, for a nonzero $p$-adic number $x$ of the form
(\ref{p-adic-number}). The field of $p$-adic numbers with the distance induced
by $\left\vert \cdot\right\vert _{p}$ is a complete ultrametric space. The
ultrametric property refers to the fact that $\left\vert x-y\right\vert
_{p}\leq\max\left\{  \left\vert x-z\right\vert _{p},\left\vert z-y\right\vert
_{p}\right\}  $ for any $x$, $y$, $z$ in $\mathbb{Q}_{p}$.

The classical Eigen-Schuster equation describes the concentration $X\left(
I,t\right)  $ of a sequence $I$ at time $t$. In the non-Archimedean approach
the sequence $I$ is codified as a $p$-adic number of the form%
\[
I=I_{-M}p^{-M}+I_{-M+1}p^{-M+1}+\cdots+I_{0}+\cdots+I_{M-1}p^{M-1},
\]
where $M\in\mathbb{N\smallsetminus}\left\{  0\right\}  $ and the set of
sequences is $G_{M}=p^{-M}\mathbb{Z}_{p}/p^{M}\mathbb{Z}_{p}$. In the limit
when $M$ tends to infinity $G_{M}$ becomes $\mathbb{Q}_{p}$. The $p$-adic
Eigen-Schuster model is given by
\begin{gather*}
\frac{d}{dt}X\left(  I,t\right)  =\frac{1}{C}\sum_{J\in G_{M}}Q\left(
I,J\right)  f\left(  J\right)  X(J,t)-\Phi_{M}\left(  t\right)  X\left(
J,t\right)  \text{, }I\in G_{M},\text{ }t>0\text{, where}\\
\Phi_{M}\left(  t\right)  =p^{-M}\sum_{I\in G_{M}}f\left(  I\right)
X(I,t),\text{ \ }%
\end{gather*}
here $\left[  Q\left(  I,J\right)  \right]  _{I,J}$ is the mutation matrix,
$f\left(  J\right)  $ is the fitness of the sequence $J$, and $C$, $p^{-M}$
are scale constants. In the limit $M$ tends to infinity, $I$ becomes a
continuous $p$-adic variable denoted as $x$, and the model takes de form%
\begin{gather*}
\frac{\partial X\left(  x,t\right)  }{\partial t}=%
{\displaystyle\int\limits_{\mathbb{Q}_{p}}}
Q(x,y)f(y)X\left(  y,t\right)  dy-\Phi\left(  t\right)  X\left(  x,t\right)
\text{, }x\in\mathbb{Q}_{p}\text{, }t\in\mathbb{R}_{+}\text{, where}\\
\Phi\left(  t\right)  =%
{\displaystyle\int\limits_{\mathbb{Q}_{p}}}
f\left(  y\right)  X\left(  y,t\right)  dy\text{\ \ for }t\geq0\text{.}%
\end{gather*}
The integral is with respect to the Haar measure of $\mathbb{Q}_{p}$. This
reasoning is not possible if we use Riemann integrals. This limit can be
formulated in a rigorous mathematical way, see e.g. \cite{zuniga-Nonlieal}.
The above model describes the time evolution of the concentration $X\left(
x,t\right)  $ of the sequence $x$, which has a\ length varying in time. A
fundamental observation is that the error threshold phenomenon occurs
independently of the topology of the space of sequences, see Section
\ref{Sub_Section_MS}.

In Section \ref{Section_3}, we introduce a very general class of models where
the mutation measure $Q(x,y,t)dy$ and the fitness function $f\left(
y,t\right)  $ depend on the time $t$. The study of the Cauchy problems for
these equations is an open problem. The case in which the fitness function is
supported in the unit ball, and the mutation measure has the form $Q(x-y)dy$
is fully studied in in this article.

We denote by $\mathbb{Z}_{p}$ the unit ball, which consists of the all the
sequences with expansions of the form (\ref{p-adic-number}) with $-k\geq0$.
The fitness landscape is given by a test function $f:$ $\mathbb{Z}%
_{p}\rightarrow\mathbb{R}_{+}$, which means that $f$ is a locally constant
function with compact support. With respect to the mutation mechanism, we only
assume the existence of a mutation measure $Q_{0}\left(  \left\vert
x\right\vert _{p}\right)  dx$, where $Q_{0}:$ $\mathbb{R}_{+}\rightarrow
\mathbb{R}_{+}$, and $dx$ is the normalized Haar measure of the group $\left(
\mathbb{Z}_{p},+\right)  $, with $\int_{\mathbb{Z}_{p}}Q_{0}\left(  \left\vert
x\right\vert _{p}\right)  dx=1$, such that the probability that a sequence $x$
mutates into a sequence \ belonging to the set $B$ is given by $\int_{B}$
$Q_{0}\left(  \left\vert x-y\right\vert _{p}\right)  dy$. \ In our model the
concentration $X\left(  x,t\right)  \in\left[  0,1\right]  $ of the sequence
$x$ at the time $t$ is controlled by the following evolution equation:%
\begin{equation}
\left\{
\begin{array}
[c]{l}%
\frac{\partial X\left(  x,t\right)  }{\partial t}=Q_{0}\left(  \left\vert
x\right\vert _{p}\right)  \ast\left\{  f\left(  \left\vert x\right\vert
_{p}\right)  X\left(  x,t\right)  \right\}  -\Phi\left(  t\right)  X\left(
x,t\right)  ,\\
\\
X\left(  x,0\right)  =X_{0}(x)\text{, }%
{\displaystyle\int\limits_{\mathbb{Z}_{p}}}
X_{0}(x)dx=1,
\end{array}
\right.  \label{EC_0}%
\end{equation}
where $x\in\mathbb{Z}_{p},t\geq0$, and $\Phi\left(  t\right)  =\int
_{\mathbb{Z}_{p}}f\left(  \left\vert y\right\vert _{p}\right)  X\left(
y,t\right)  dy$. The term
\[
\boldsymbol{W}_{0}X\left(  x,t\right)  =Q_{0}\left(  \left\vert x\right\vert
_{p}\right)  \ast\left\{  f\left(  \left\vert x\right\vert _{p}\right)
X\left(  x,t\right)  \right\}
\]
represents the rate at which the sequences are mutating into the sequence $x$.
We assume that the replication reactions occur in a chemostat, see e.g.
\cite{Tannembaum et al}, which is a device that allows the maintenance of a
constant population size, this mechanism is implemented by using the term
$-\Phi\left(  t\right)  X\left(  x,t\right)  $.

We now discuss briefly the main results presented in this article. To study
the Cauchy problem (\ref{EC_0}), we first construct an function space
invariant under operator $\boldsymbol{W}_{0}$, and then use $p$-adic wavelets
to construct an explicit solution of the Cauchy problem using the classical
method of separation of variables. The construction of the invariant space for
$\boldsymbol{W}_{0}$ \ was not considered in \cite{Zuniga-JPhA}. We assume
that the fitness function is a test function of the form $f(x)=\sum_{J\in
G_{M}}f(J)\Omega\left(  p^{M}\left\vert x-J\right\vert _{p}\right)  $, where
$f(J)>0$, and $\Omega\left(  p^{M}\left\vert x-J\right\vert _{p}\right)  $ is
the characteristic function of the ball $J+p^{M}\mathbb{Z}_{p}$ and
$\mathbb{Z}_{p}=%
{\textstyle\bigsqcup\nolimits_{J\in G_{M}}}
\left(  J+p^{M}\mathbb{Z}_{p}\right)  $. The ball $J+p^{M}\mathbb{Z}_{p}$ is a
cloud of mutants around the master sequence $J$, any sequence in this cloud
reproduces at a rate of $f(J)$ copies per unit of time. We denote by
$\mathcal{D}_{M}$ the $\mathbb{R}$-vector space spanned by the functions
$\Omega\left(  p^{M}\left\vert x-I\right\vert _{p}\right)  ,$ $I\in G_{M}$.
The finite dimensional vector space $\mathcal{D}_{M}$ is invariant under
$\boldsymbol{W}_{0}$, then its restriction to $\mathcal{D}_{M}$ represented by
a matrix $\mathbb{W}^{0}=\left[  \mathbb{W}_{I,J}^{0}\right]  _{I,J\in G_{M}}%
$. We denote by $L^{2}\left(  J+p^{M}\mathbb{Z}_{p}\right)  $ the $\mathbb{C}%
$-vector space of square-integrable functions defined on the ball
$J+p^{M}\mathbb{Z}_{p}$. The space $L(\boldsymbol{W}_{0})=\mathcal{D}_{M}%
{\textstyle\bigoplus}
{\textstyle\bigoplus_{J\in G_{M}}}
L^{2}\left(  J+p^{M}\mathbb{Z}_{p}\right)  $ is invariant under
$\boldsymbol{W}_{0}$, see Lemma \ref{Lemma_1}.

We solve the Cauchy problem (\ref{EC_0}) in $L(\boldsymbol{W}_{0})$:%
\begin{gather}
X\left(  x,t\right)  =\frac{\left(  e^{t\mathbb{W}^{0}}\left[  C_{I}%
^{0}(0)\right]  _{I\in G_{M}}\right)  \left[  \Omega\left(  p^{M}\left\vert
x-I\right\vert _{p}\right)  \right]  _{I\in G_{M}}^{T}}{\overline{Y\left(
t\right)  }}+\nonumber\\%
{\displaystyle\sum\limits_{I\in G_{M}}}
\text{\ }\sum\limits_{\text{supp}\Psi_{rnj}\subseteq I+p^{M}\mathbb{Z}_{p}%
}\frac{\text{ }e^{t\widehat{Q_{0}}\left(  p^{1-r}\right)  f(I)}%
\operatorname{Re}\left(  C_{rjn}^{I}\left(  0\right)  \Psi_{rnj}\left(
x\right)  \right)  }{\overline{Y\left(  t\right)  }},\label{Solution}%
\end{gather}
where $\Psi_{rnj}\left(  x\right)  $s are wavelet basis of $L^{2}\left(
\mathbb{Q}_{p}\right)  $, each of these functions has average zero, i.e.
\begin{equation}%
{\displaystyle\int\limits_{\mathbb{Q}_{p}}}
\Psi_{rnj}\left(  x\right)  dx=0,\label{average_zer}%
\end{equation}
and $\widehat{Q_{0}}\left(  \xi\right)  $ is the Fourier transform of the
radial function $Q_{0}\left(  \left\vert x\right\vert _{p}\right)  $. Now the
initial datum $X\left(  x,0\right)  $ is determined by an element of the set
of sequences%
\[
\mathcal{S}=%
{\textstyle\bigsqcup\limits_{J\in G_{M}}}
\left\{  C_{J}^{0}(0)\right\}
{\textstyle\bigsqcup}
{\textstyle\bigsqcup\limits_{I\in G_{M}}}
\left\{  C_{rjn}^{I}\left(  0\right)  \right\}  _{rnj}.
\]
If we take the initial datum determined by the conditions $C_{rjn}^{I}\left(
0\right)  =0$ for any $I$, $rnj$ and $C_{J}^{0}(0)\neq0$ for some $J$, i.e.
$X\left(  x,0\right)  \in\mathcal{D}_{M}$, then the condition
(\ref{average_zer}) implies that%
\[
X\left(  x,t\right)  =\frac{\left(  e^{t\mathbb{W}^{0}}\left[  C_{I}%
^{0}(0)\right]  _{I\in G_{M}}\right)  \left[  \Omega\left(  p^{M}\left\vert
x-I\right\vert _{p}\right)  \right]  _{I\in G_{M}}^{T}}{\overline{Y\left(
t\right)  }},
\]
see Theorem \ref{Theorem_0}. A key observation is that given $t\geq0$, the
value of the function $X\left(  x,t\right)  $ depends only on the first $M$
digits of $x=x_{0}+x_{1}p+\ldots+x_{M-1}p^{M-1}+\ldots$, since $M$ is fixed,
in this model the length of the sequences does not change in time. Since
$\mathbb{W}^{0}$ is a real symmetric matrix, it is diagonalizable and all its
eigenvalues are real. In this case $\lim_{t\rightarrow\infty}X\left(
x,t\right)  $ exists and it is controlled by the largest eigenvalue of
$\mathbb{W}^{0}$. This situation corresponds to \textit{the survival of the
fitter}. This is the typical scenario predicted by the classical
Eigen-Schuster model. In this scenario the Eigen paradox may happen. This
result says that the classical Eigen-Schuster description of evolution can be
obtained using the $p$-adic Eigen-Schuster \ model. Our previous publication
\cite{Zuniga-JPhA} does not contain a similar result.

To escape to the Eigen paradox the length of the sequences must growth, see
Section \ref{Sub_Section_MS}, which requires that some of the oscillatory
terms (those involving the $\Psi_{rnj}\left(  x\right)  $s) must be preserved
in the\ long term in (\ref{Solution}). If these oscillatory terms do not
vanish in the long term, we have a cloud of sequences in $\mathbb{Z}_{p}$
(with time variable length, including sequences of infinite length) evolving
according to the basic Darwinian principles. Intuitively, one must show that
the function $\lim_{t\rightarrow\infty}X\left(  x,t\right)  $ depends on
infinitely many digits in the $p$-adic expansion of $x$. This is a non-trivial
mathematical task that requires suitable hypotheses on the mutation measure,
and surprisingly non-trivial results on stochastic processes on $\mathbb{Q}%
_{p}$. In our previous publication \cite{Zuniga-JPhA} these matters were not considered.

We pick as a mutation measure a family of Gibbs type measures of the form
$Q_{0}\left(  \left\vert x\right\vert _{p}\right)  =\mathcal{N}\Omega\left(
\left\vert x\right\vert _{p}\right)  \exp(-\sigma\left\vert x\right\vert
_{p}^{\alpha})$, where $\alpha$, $\sigma$ are positive parameters,
$\Omega\left(  \left\vert x\right\vert _{p}\right)  $ is the characteristic
function of the unit ball, and $\mathcal{N}$ is a normalization constant. The
Fourier transform $Q_{0}\left(  \left\vert x\right\vert _{p}\right)  $ of
satisfies $\mathcal{F}_{x\rightarrow\xi}(Q_{0}\left(  \left\vert x\right\vert
_{p};\sigma,\alpha\right)  )=\mathcal{N}Z(\xi;\sigma,\alpha)\ast\Omega\left(
\left\vert \xi\right\vert _{p}\right)  $, where $Z(\xi;\sigma,\alpha)$ is the
classical $p$-adic heat kernel, which is the transition density of a Markov
process in $\mathbb{Q}_{p}$, see e.g. \cite{Bendikov et al}, \cite{KKZuniga},
\cite{V-V-Z}, \cite{Zuniga-LNM-2016}. We use the extensively the results about
the behavior of $Z(\xi;\sigma,\alpha)$ around the origin and at the infinity.

We introduce the following two conditions:%
\[
\text{Hypothesis A:\hspace{0.8in}}%
{\displaystyle\int\limits_{p^{M}\mathbb{Z}_{p}}}
Q_{0}\left(  \left\vert z\right\vert _{p}\right)  dz\in\left(  \frac{1}%
{2},1\right)  .
\]%
\[
\text{Hypothesis B:\hspace{0.8in} }\widehat{Q_{0}}\left(  p^{1-r_{0}}\right)
f(I_{0})>\mu_{\max},
\]
where $\mu_{\max}$ is the largest eigenvalue of $\mathbb{W}^{0}$. The
hypothesis A says that the probability that a sequence belonging to
$I+p^{M}\mathbb{Z}_{p}$ mutates into a sequence belonging to $%
{\textstyle\bigsqcup\nolimits_{J\neq I}}
\left(  J+p^{M}\mathbb{Z}_{p}\right)  $ is less than $\frac{1}{2}$, for any
$I\in G_{M}$. It is important to mention here that Hypothesis A also appears
in the Maynard Smith ansatz, see \cite{Maynad Smith}, \cite{Szat-PTRSB},
\cite{Zuniga-JPhA}.

We denote by $\mathcal{S}_{0}$ the subset of $\mathcal{S}$ consists of the
sequences satisfying $C_{rjn}^{I}\left(  0\right)  \neq0$ for some $I$, $rnj$
and $C_{J}^{0}(0)\neq0$ for some $J$. We show the existence of $\sigma_{\max}$
and $M=M(\sigma,\alpha)$, such that if $\ X(x,0)\in\mathcal{S}_{0}$, and
$Q_{0}\left(  \left\vert x\right\vert _{p};\sigma,\alpha\right)  $ satisfies
that $\sigma\in\left(  0,\sigma_{\max}\right)  $, $\alpha\in\left(
0,\infty\right)  $, and $\int_{p^{M}\mathbb{Z}_{p}}Q_{0}\left(  \left\vert
x\right\vert _{p}\right)  dx\in\left(  \frac{1}{2},1\right)  $, then the
Cauchy problem (\ref{EC_0}) has a solution with a non-trivial oscillatory
behavior at infinity. In this case, we say that (\ref{EC_0}) admits a
\textit{quasispecies solution}, see Theorems \ref{Theorem_D}, \ref{Theorem_E}.
A quasispecies is a large group of related genotypes that exist in an
environment of high mutation rate at stationary state, where a large fraction
of offspring are expected to contain one or more mutations relative to the
parent, \cite{Eigen et al}. The $p$-adic quasispecies correspond to a profile
of a solution of the Cauchy problem (\ref{eq6}) when $t$ tends to infinity.

As a generalization of the classical Eigen-Schuster model, our $p$-adic model
encodes the basic principles of Darwinian evolution. The long term survival of
a sequence under the selection pressure depends on the interaction of the
fitness function, the mutation measure and the initial concentration of the
sequences. Assuming that the mutation measure is a Gibbs measure of type
$\Omega\left(  \left\vert x\right\vert _{p}\right)  \exp(-\sigma\left\vert
x\right\vert _{p}^{\alpha})$, we establish the existence of scenarios where
the long-term concentration of sequences with arbitrary length does not
vanish. Based on the Maynard Smith ansatz, see Section \ref{Sub_Section_MS},
we interpret this situation as that long-term survival requires that the
complexity of the genomes grow at the right pace. This agrees with the fact
that larger genome size improves the replication fidelity. In other case, the
Eigen paradox occurs, which means that the sequences are unable to copy
themselves with sufficiently fidelity, and thus in the long term these
sequences will not survive. The Eigen paradox is a possible scenario, among
infinitely many, in our $p$-adic evolution model.

It is important to mention here that \ Avetisov and Zhuravlev pointed out
using $1D$ $p$-adic diffusion equation in biological evolution, see
\cite{Av-Zhu}-\cite{Av-Zhu-2}. This approach does not allow to analyze
directly the error catastrophe in the standard sense. On the other hand, the
use of $p$-adic numbers in DNA models and analysis of the genetic code is
well-known see e.g. \cite{Dra-Kh-K-V}, \cite{Dragovich}, \cite{KKGentic}, and
the references therein.

The article is organized a s follows. In Section \ref{Section_2} we review the
essential ideas about $p$-adic analysis. In Section \ref{Section_3}, we review
and extend the $p$-adic version of Eigen-Schuster model introduced in
\cite{Zuniga-JPhA}. In the extended models the mutation measure is a
transition density function $Q\left(  x,y,t\right)  \geq0$ for $x$,
$y\in\mathbb{Q}_{p}$, $t>0$, \ of a Markov process. The study of the Cauchy
problem for these models is an open problem. In Section \ref{Section_4}, we
study the existence of a solution for the Cauchy problem considered in the
introduction. We use the classical method of separation of variables and
$p$-adic wavelets. In Section \ref{Section_5}, we show the existence of
$p$-adic quasispecies and discuss the solution of Eigen's paradox.

\section{\label{Section_2}$p$-Adic Analysis: Essential Ideas}

In this Section, we collect some basic results on $p$-adic analysis that we
use through the article. For a detailed exposition the reader may consult
\cite{Alberio et al}, \cite{Kochubei}, \cite{Taibleson}, \cite{V-V-Z}.

\subsection{The field of $p$-adic numbers}

Along this article $p$ will denote a prime number. The field of $p-$adic
numbers $%
\mathbb{Q}
_{p}$ is defined as the completion of the field of rational numbers
$\mathbb{Q}$ with respect to the $p-$adic norm $|\cdot|_{p}$, which is defined
as
\[
\left\vert x\right\vert _{p}=\left\{
\begin{array}
[c]{lll}%
0 & \text{if} & x=0\\
&  & \\
p^{-\gamma} & \text{if} & x=p^{\gamma}\frac{a}{b}\text{,}%
\end{array}
\right.
\]
where $a$ and $b$ are integers coprime with $p$. The integer $\gamma:=ord(x)
$, with $ord(0):=+\infty$, is called the\textit{\ }$p-$\textit{adic order of}
$x$.

Any $p-$adic number $x\neq0$ has a unique expansion of the form
\[
x=p^{ord(x)}\sum_{j=0}^{\infty}x_{j}p^{j},
\]
where $x_{j}\in\{0,\dots,p-1\}$ and $x_{0}\neq0$. By using this expansion, we
define \textit{the fractional part of }$x\in\mathbb{Q}_{p}$, denoted
$\{x\}_{p}$, as the rational number
\[
\left\{  x\right\}  _{p}=\left\{
\begin{array}
[c]{lll}%
0 & \text{if} & x=0\text{ or }ord(x)\geq0\\
&  & \\
p^{ord(x)}\sum_{j=0}^{-ord_{p}(x)-1}x_{j}p^{j} & \text{if} & ord(x)<0.
\end{array}
\right.
\]

For $r\in\mathbb{Z}$, denote by $B_{r}(a)=\{x\in%
\mathbb{Q}
_{p};\left\vert x-a\right\vert _{p}\leq p^{r}\}$ \textit{the ball of radius
}$p^{r}$ \textit{with center at} $a\in%
\mathbb{Q}
_{p}$, and take $B_{r}(0):=B_{r}$. The ball $B_{0}$ equals $\mathbb{Z}_{p}$,
\textit{the ring of }$p-$\textit{adic integers of }$%
\mathbb{Q}
_{p}$. We also denote by $S_{r}(a)=\{x\in\mathbb{Q}_{p};|x-a|_{p}=p^{r}\}$
\textit{the sphere of radius }$p^{r}$ \textit{with center at} $a\in%
\mathbb{Q}
_{p}$, and take $S_{r}(0):=S_{r}$. We notice that $S_{0}^{1}=\mathbb{Z}%
_{p}^{\times}$ (the group of units of $\mathbb{Z}_{p}$). The balls and spheres
are both open and closed subsets in $%
\mathbb{Q}
_{p}$. In addition, two balls in $%
\mathbb{Q}
_{p}$ are either disjoint or one is contained in the other.

The metric space $\left(
\mathbb{Q}
_{p},\left\vert \cdot\right\vert _{p}\right)  $ is a complete ultrametric
space. As a topological space $\left(
\mathbb{Q}
_{p},|\cdot|_{p}\right)  $ is totally disconnected, i.e. the only connected
\ subsets of $%
\mathbb{Q}
_{p}$ are the empty set and the points. In addition, $\mathbb{Q}_{p}$\ is
homeomorphic to a Cantor-like subset of the real line, see e.g. \cite{Alberio
et al}, \cite{V-V-Z}. A subset of $\mathbb{Q}_{p}$ is compact if and only if
it is closed and bounded in $\mathbb{Q}_{p}$, see e.g. \cite[Section
1.3]{V-V-Z}, or \cite[Section 1.8]{Alberio et al}. The balls and spheres are
compact subsets. Thus $\left(
\mathbb{Q}
_{p},|\cdot|_{p}\right)  $ is a locally compact topological space.

\begin{notation}
We will use $\Omega\left(  p^{-r}|x-a|_{p}\right)  $ to denote the
characteristic function of the ball $B_{r}(a)$. For more general sets, we
denote by $1_{A}$ the characteristic function of $A$.
\end{notation}

\subsection{The Haar measure}

Since $(\mathbb{Q}_{p},+)$ is a locally compact topological group, there
exists a Borel measure $dx$, called the Haar measure of $(\mathbb{Q}_{p},+)$,
unique up to multiplication by a positive constant, such that $\int_{U}dx>0$
for every non-empty Borel open set $U\subset\mathbb{Q}_{p}$, and satisfying
$\int_{E+z}dx=\int_{E}dx$ for every Borel set $E\subset\mathbb{Q}_{p}$, see
e.g. \cite[Chapter XI]{Halmos}. If we normalize this measure by the condition
$\int_{\mathbb{Z}_{p}}dx=1$, then $dx$ is unique. From now on we denote by
$dx$ the normalized Haar measure of $(\mathbb{Q}_{p},+)$.

\subsection{Some function spaces}

A complex-valued function $\varphi$ defined on $%
\mathbb{Q}
_{p}$ is \textit{called locally constant} if for any $x\in%
\mathbb{Q}
_{p}$ there exist an integer $l(x)\in\mathbb{Z}$ such that
\begin{equation}
\varphi(x+x^{\prime})=\varphi(x)\text{ for }x^{\prime}\in B_{l(x)}%
.\label{local_constancy_parameter}%
\end{equation}
\ A function $\varphi:%
\mathbb{Q}
_{p}\rightarrow\mathbb{C}$ is called a \textit{Bruhat-Schwartz function (or a
test function)} if it is locally constant with compact support. In this case,
we can take $l=l(\varphi)$ in (\ref{local_constancy_parameter}) independent of
$x$. The largest of such integers is called \textit{the parameter of local
constancy} of $\varphi$. The $\mathbb{C}$-vector space of Bruhat-Schwartz
functions is denoted by $\mathcal{D}:=\mathcal{D}(%
\mathbb{Q}
_{p},\mathbb{C})$. We will denote by $\mathcal{D}_{\mathbb{R}}:=\mathcal{D}(%
\mathbb{Q}
_{p},\mathbb{R})$, the $\mathbb{R}$-vector space of test functions.

Given $\rho\in\lbrack1,\infty)$, we denote by $L^{\rho}:=L^{\rho}\left(
\mathbb{Q}
_{p}\right)  :=L^{\rho}\left(
\mathbb{Q}
_{p},dx\right)  ,$ the $%
\mathbb{C}
-$vector space of all the complex valued functions $g$ satisfying $\int_{%
\mathbb{Q}
_{p}}\left\vert g\left(  x\right)  \right\vert ^{\rho}dx<\infty$. The
corresponding $\mathbb{R}$-vector spaces are denoted as $L_{\mathbb{R}}^{\rho
}\allowbreak:=L_{\mathbb{R}}^{\rho}\left(
\mathbb{Q}
_{p}\right)  =L_{\mathbb{R}}^{\rho}\left(
\mathbb{Q}
_{p},dx\right)  $, $1\leq\rho<\infty$.

\subsection{Fourier transform}

Set $\chi_{p}(y)=\exp(2\pi i\{y\}_{p})$ for $y\in%
\mathbb{Q}
_{p}$. The map $\chi_{p}(\cdot)$ is an additive character on $%
\mathbb{Q}
_{p}$, i.e. a continuous map from $\left(
\mathbb{Q}
_{p},+\right)  $ into $S$ (the unit circle considered as multiplicative group)
satisfying $\chi_{p}(x_{0}+x_{1})=\chi_{p}(x_{0})\chi_{p}(x_{1})$,
$x_{0},x_{1}\in%
\mathbb{Q}
_{p}$. The additive characters of $%
\mathbb{Q}
_{p}$ form an Abelian group which is isomorphic to $\left(
\mathbb{Q}
_{p},+\right)  $, the isomorphism is given by $\xi\rightarrow\chi_{p}(\xi x)$,
see e.g. \cite[Section 2.3]{Alberio et al}.

If $f\in L^{1}$ its Fourier transform is defined by
\[
(\mathcal{F}f)(\xi)=%
{\displaystyle\int\limits_{\mathbb{Q}_{p}}}
\chi_{p}(\xi x)f(x)dx,\quad\text{for }\xi\in%
\mathbb{Q}
_{p}.
\]
We will also use the notation $\mathcal{F}_{x\rightarrow\xi}f$ and
$\widehat{f}$\ for the Fourier transform of $f$. The Fourier transform is a
linear isomorphism from $\mathcal{D}$ onto itself satisfying
\begin{equation}
(\mathcal{F}(\mathcal{F}f))(\xi)=f(-\xi), \label{FF(f)}%
\end{equation}
for every $f\in\mathcal{D},$ see e.g. \cite[Section 4.8]{Alberio et al}. If
$f\in L^{2}$, its Fourier transform is defined as
\[
(\mathcal{F}f)(\xi)=\lim_{k\rightarrow\infty}%
{\displaystyle\int\limits_{|x|_{p}\leq p^{k}}}
\chi_{p}(\xi\cdot x)f(x)d^{n}x,\quad\text{for }\xi\in%
\mathbb{Q}
_{p}\text{,}%
\]
where the limit is taken in $L^{2}$. We recall that the Fourier transform is
unitary on $L^{2},$ i.e. $||f||_{L^{2}}=||\mathcal{F}f||_{L^{2}}$ for $f\in
L^{2}$ and that (\ref{FF(f)}) is also valid in $L^{2}$, see e.g. \cite[Chapter
III, Section 2]{Taibleson}.

\section{\label{Section_3}$p$-Adic models of Eigen-Schuster type}

In this section we review and extend the $p$-adic version of Eigen-Schuster
model introduced in \cite{Zuniga-JPhA}, see e.g. \cite{Eigen1971}, \cite{Eigen
et al}, \cite{Nowak}, \cite{SchusterPeter}, \cite{Szat-PTRSB},
\cite{Tannembaum et al} for the classical model. This model describes
mutation-selection process of replicating sequences, when the sequences are
represented by $p$-adic numbers.

\subsection{The Model}

\subsubsection{The space of sequences}

A replicator is a model of an entity with the template property, which means
that it serves as a pattern for the generation of another replicator. This
copying process is subject to errors (mutations). Along this article we use
replicators, genomes and sequences as synonyms. The assumption of the
existence of replicators implies that the information stored in the
replicators is modified randomly, and that part of it is fixed due to the
selection pressure, which in turn is related with the self-replicate capacity
of the replicators (their fitness).

In our model each sequence corresponds to a $p$-adic number:%
\[
x=x_{-m}p^{-m}+x_{-m+1}p^{-m+1}+\ldots+x_{0}+x_{1}p+\ldots
\]
where the digits $x_{i}$s run through the set $\left\{  0,1,\ldots
,p-1\right\}  $. Consequently, in our model the sequences are words of
arbitrary length written in the alphabet $0$, $1$,$\ldots$, $p-1$, and the
space of sequences is $\left(  \mathbb{Q}_{p},\left\vert \cdot\right\vert
_{p}\right)  $, which is an infinite ultrametric space.

\subsubsection{Concentrations}

The concentration $X\left(  x,t\right)  $ of sequence $x\in\mathbb{Q}_{p}$ at
the time $t\geq0$ is a \ real number between zero and one. In addition, we
assume that%
\begin{equation}%
{\displaystyle\int\limits_{\mathbb{Q}_{p}}}
X\left(  y,t\right)  dy=1\text{ for }t>0\text{.} \label{hypotesis_1}%
\end{equation}
This last condition assures that the total concentration remains constant for
$t>0$.

\subsubsection{The mutation measure}

We fix a transition density function $Q\left(  x,y,t\right)  \geq0$ for $x$,
$y\in\mathbb{Q}_{p}$, $t>0$, which means that given a Borel subset
$E\subseteq\mathbb{Q}_{p}$,%
\[%
{\displaystyle\int\limits_{E}}
Q\left(  x,y,t\right)  dy
\]
represents the probability that the sequence $x$ will mutate into a sequence
belonging to set $E$ at time $t$. We assume that%
\[%
{\displaystyle\int\limits_{\mathbb{Q}_{p}}}
Q\left(  x,y,t\right)  dy=1\text{ for any }x\in\mathbb{Q}_{p}\text{ and }t>0.
\]

We call $Q\left(  x,y,t\right)  dy$ \textit{a mutation measure}. We set
$\mathbb{R}_{+}=\left\{  y\in\mathbb{R};y\geq0\right\}  $. A simple\ way of
constructing time independent mutation measures is as follows. Take
$Q:\mathbb{R}_{+}\mathbb{\rightarrow R}_{+}$, and set $Q(\left\vert
x\right\vert _{p})$ such that
\[
\text{ }%
{\displaystyle\int\limits_{\mathbb{Q}_{p}}}
Q\left(  \left\vert y\right\vert _{p}\right)  dy=1.
\]
Then, for a Borel set $E\subseteq\mathbb{Q}_{p}$ and $x\in\mathbb{Q}_{p}$, the
integral
\[%
{\displaystyle\int\limits_{E}}
Q\left(  \left\vert x-y\right\vert _{p}\right)  dy
\]
gives the probability that sequence $x$ will mutate into a sequence belonging
to $E$.

\subsubsection{The\textit{ }fitness function}

The\textit{ fitness function} $f$ $\left(  x,t\right)  $, $x\in\mathbb{Q}_{p}%
$, $t>0$, is a non-negative bounded function. The simplest choice for $f$ is
test function independent of the time. This case was considered in
\cite{Zuniga-JPhA}. The assumption that function $f$ \ has compact support
means that the evolution process is limited to a certain region of the space
of sequences, which is infinite.

\subsubsection{The non-Archimedean replicator equation}

For $t>0$ fixed, and $x\in\mathbb{Q}_{p}$, we set%
\[
\left(  \boldsymbol{W}\varphi\right)  \left(  x\right)  =%
{\displaystyle\int\limits_{\mathbb{Q}_{p}}}
Q\left(  x,y,t\right)  \left\{  f\left(  y,t\right)  \varphi\left(  y\right)
\right\}  dy.
\]
Under the hypotheses:%
\[
f\left(  x,t\right)  \leq C\text{ for any }x\in\mathbb{Q}_{p},
\]
where $C$ is a positive constant, and
\[
Q\left(  x,\cdot,t\right)  \in L^{2},
\]
we have $\boldsymbol{W}:L^{2}(\mathbb{Q}_{p})\rightarrow L^{2}(\mathbb{Q}%
_{p})$ is a well-defined continuous operator.

Our non-Archimedean Eigen-Schuster models have the form:%
\begin{equation}
\frac{\partial X\left(  x,t\right)  }{\partial t}=\boldsymbol{W}X\left(
x,t\right)  -\Phi\left(  t\right)  X\left(  x,t\right)  \text{, }%
x\in\mathbb{Q}_{p}\text{, }t\in\mathbb{R}_{+}\text{,} \label{eq1}%
\end{equation}
where%
\begin{equation}
\Phi\left(  t,X\right)  :=\Phi\left(  t\right)  =%
{\displaystyle\int\limits_{\mathbb{Q}_{p}}}
f\left(  y,t\right)  X\left(  y,t\right)  dy\text{\ \ for }t\geq0\text{.}
\label{eq1A}%
\end{equation}
This function $\Phi\left(  t\right)  $\ is used to maintain constant the total
concentration in the chemostat.

\subsection{Discretization}

We fix $M\in\mathbb{N\smallsetminus}\left\{  0\right\}  $ and set
\[
G_{M}:=p^{-M}\mathbb{Z}_{p}/p^{M}\mathbb{Z}_{p}.
\]
We consider $G_{M}$\ as an additive group and fix the following systems of
representatives:%
\begin{equation}
I=I_{-M}p^{-M}+I_{-M+1}p^{-M+1}+\cdots+I_{0}+\cdots+I_{M-1}p^{M-1},
\label{seq_I}%
\end{equation}
where the $I_{j}$s belong to $\left\{  0,1,\ldots,p-1\right\}  $. Furthermore,
the restriction of $\left\vert \cdot\right\vert _{p}$ to $G_{M}$ induces an
absolute value such that $\left\vert G_{M}\right\vert _{p}=\left\{
0,p^{-\left(  M+1\right)  },\cdots,p^{-1},1,\cdots,p^{M}\right\}  $. We endow
$G_{M}$ with the metric induced by $\left\vert \cdot\right\vert _{p}$, and
thus $G_{M}$ becomes a finite ultrametric space. In addition, $G_{M}$ can be
identified with the set of branches (vertices at the top level) of a rooted
tree with $2M+1$ levels and $p^{2M}$ branches.

We denote by $\mathcal{D}_{M}$ the $\mathbb{R}$-vector subspace of
$\mathcal{D}_{\mathbb{R}}$ spanned by the functions%
\[
\Omega\left(  p^{M}\left\vert x-I\right\vert _{p}\right)  ,\text{ }I\in
G_{M}\text{.}%
\]
Notice that $\Omega\left(  p^{M}\left\vert x-I\right\vert \right)
\Omega\left(  p^{M}\left\vert x-J\right\vert \right)  =0$ for any $x$, if
$I\neq J$. Thus, any function $\varphi\in\mathcal{D}_{M}$ has the form%
\[
\varphi\left(  x\right)  =\sum_{I\in G_{M}}\varphi\left(  I\right)
\Omega\left(  p^{M}\left\vert x-I\right\vert _{p}\right)  ,
\]
where the $\varphi\left(  I\right)  $s are real numbers. The dimension of
$\mathcal{D}_{M}\left(  \mathbb{Q}_{p}\right)  $ is $\#G_{M}=p^{2M}$.

In order to explain the connection between the non-Archimedean replicator
equation (\ref{eq1}) and the classical one, we assume that $f\left(
\cdot,t\right)  $ and $X(\cdot,t)$ belong to $\mathcal{D}_{M}$, and that
$Q\left(  \cdot,\cdot,t\right)  $ belongs to $\mathcal{D}_{M}\times
\mathcal{D}_{M}$ for any $t$. This assumption means that the mentioned
functions can be very well approximated by functions in $\mathcal{D}_{M}$,
respectively in $\mathcal{D}_{M}\times\mathcal{D}_{M}$, see e.g.
\cite{zuniga-Nonlieal}. Then%
\begin{gather*}
Q\left(  x,y,t\right)  =\\
\frac{1}{C_{M}}\sum_{J\in G_{M}}\sum_{I\in G_{M}}Q\left(  I,J,t\right)
\Omega\left(  p^{M}\left\vert x-I\right\vert _{p}\right)  \Omega\left(
p^{M}\left\vert y-J\right\vert _{p}\right) \\
=\frac{1}{C_{M}}\sum_{I\in G_{M}}\left\{  \sum_{J\in G_{M}}Q\left(
I,J,t\right)  \Omega\left(  p^{M}\left\vert y-J\right\vert _{p}\right)
\right\}  \Omega\left(  p^{M}\left\vert x-I\right\vert _{p}\right)  ,
\end{gather*}
where $C_{M}=p^{-M}\sum_{J\in G_{M}}Q\left(  I,J,t\right)  $, and the
$Q\left(  I,J,t\right)  $s are real-valued functions of class $C^{1}$ in $t$,
\[
f\left(  x,t\right)  =\sum_{J\in G_{M}}f\left(  J,t\right)  \Omega\left(
p^{M}\left\vert x-J\right\vert _{p}\right)  ,
\]
and%
\[
X(x,t)=\sum_{J\in G_{M}}X(J,t)\Omega\left(  p^{M}\left\vert x-J\right\vert
_{p}\right)  \text{ for any }t\geq0\text{,}%
\]
where each $X(I,t)$ is a real-valued function of class $C^{1}$ in $t$. Now%
\begin{multline*}%
{\displaystyle\int\limits_{\mathbb{Q}_{p}}}
Q\left(  x,y,t\right)  f\left(  y,t\right)  X\left(  y,t\right)  dy=\\
\frac{1}{C_{M}}\sum_{I\in G_{M}}\left\{  \sum_{J\in G_{M}}Q\left(
I,J,t\right)  f\left(  J,t\right)  X(J,t)%
{\displaystyle\int\limits_{\mathbb{Q}_{p}}}
\Omega\left(  p^{M}\left\vert y-J\right\vert _{p}\right)  dy\right\}
\Omega\left(  p^{M}\left\vert x-I\right\vert _{p}\right) \\
=\frac{1}{C}\sum_{I\in G_{M}}\left\{  \sum_{J\in G_{M}}Q\left(  I,J,t\right)
f\left(  J,t\right)  X(J,t)\right\}  \Omega\left(  p^{M}\left\vert
x-I\right\vert _{p}\right)  ,
\end{multline*}
where $C=p^{-M}\sum_{J\in G_{M}}Q\left(  I,J,t\right)  $. Finally, using the
fact that the $\Omega\left(  p^{M}\left\vert x-I\right\vert _{p}\right)  $,
$I\in G_{M}$ are $\mathbb{R}$-linearly independent, we get%
\begin{equation}
\frac{d}{dt}X\left(  I,t\right)  =\frac{1}{C}\sum_{J\in G_{M}}Q\left(
I,J,t\right)  f\left(  J,t\right)  X(J,t)-\Phi_{M}\left(  t\right)  X\left(
J,t\right)  ,\text{ \ } \label{eq_Eigen}%
\end{equation}
for $I\in G_{M}$, where
\begin{equation}
\Phi_{M}\left(  t\right)  =p^{-M}\sum_{I\in G_{M}}f\left(  I,t\right)  X(I,t).
\label{eq_Eigen_1}%
\end{equation}
In the case in which the $Q\left(  I,J,t\right)  $s are independent of the
time, (\ref{eq_Eigen})-(\ref{eq_Eigen_1}) is the classical Eigen-Schuster
model on the finite ultrametric space $G_{M}$. In \cite{Zuniga-JPhA}, see also
\cite{zuniga-Nonlieal}, we argue that the system (\ref{eq1})-(\ref{eq1A}) is
the limit when $M$ tends to infinity of the system (\ref{eq_Eigen}%
)-(\ref{eq_Eigen_1}).

\subsection{\label{Sub_Section_MS}The error threshold and the topology of the
space of sequences}

In this section we review the Maynard Smith approach to the error threshold
problem assuming that the space of sequence is an arbitrary measurable metric
space $(\mathbb{Y},\mu)$, see \cite{Maynad Smith}, \cite{Szat-PTRSB} for the
classical version. This means that we do not assume a specific topology for
the space of sequences, in particular, the length of the sequences is
arbitrary. We divide the space of sequences into two disjoint sets:
\begin{equation}
\mathbb{Y}=\mathbb{A}\bigsqcup\mathbb{B}, \label{condition_0}%
\end{equation}
and assume that
\begin{equation}
f\mid_{\mathbb{A}}\equiv a\text{, \ \ }f\mid_{\mathbb{B}}\equiv b\text{,
\ with }a>b\text{,} \label{condition_1}%
\end{equation}
here \textquotedblleft$\equiv$\textquotedblright\ means identically equal. We
denote by $X(x,t)$ the concentration of sequences of type $\mathbb{A}$ and by
$Y(x,t)$ the concentration of sequences of type $\mathbb{B}$. Notice that the
supports of $X(x,t)$\ and $Y(x,t)$\ are disjoint. We denote by $q$ the
probability that a sequence in $\mathbb{A}$ mutates into a sequence belonging
to $\mathbb{B}$, and by $r$ the probability of mutation of a sequence from
$\mathbb{B}$ into a sequence in $\mathbb{A}$. The system of equations
governing the development of these populations is
\begin{equation}%
\begin{array}
[c]{cc}%
\frac{\partial X(x,t)}{\partial t}= & a\left(  1-q\right)
X(x,t)+brY(x,t)-\Phi\left(  t\right)  X\left(  x,t\right) \\
& \\
\frac{\partial Y(x,t)}{\partial t}= & aqX(x,t)+b\left(  1-r\right)
Y(x,t)-\Phi\left(  t\right)  Y\left(  x,t\right)  ,
\end{array}
\label{system1}%
\end{equation}
where
\[
\int\limits_{\mathbb{A}}X\left(  x,t\right)  d\mu\left(  x\right)
+\int\limits_{\mathbb{B}}Y\left(  x,t\right)  d\mu\left(  x\right)  =1,
\]
and
\begin{align*}
\Phi\left(  t\right)   &  =\int\limits_{\mathbb{A}}f\left(  x\right)  X\left(
x,t\right)  d\mu\left(  x\right)  +\int\limits_{\mathbb{B}}f\left(  x\right)
Y\left(  x,t\right)  d\mu\left(  x\right) \\
&  =a\int\limits_{\mathbb{A}}X\left(  x,t\right)  d\mu\left(  x\right)
+b\int\limits_{\mathbb{B}}Y\left(  x,t\right)  d\mu\left(  x\right)  .
\end{align*}
We assume that \ $r$ is very small, so we can assume that system
(\ref{system1}) has the form%
\begin{equation}%
\begin{array}
[c]{cc}%
\frac{\partial X(x,t)}{\partial t}= & a\left(  1-q\right)  X(x,t)-\Phi\left(
t\right)  X\left(  x,t\right) \\
& \\
\frac{\partial Y(x,t)}{\partial t}= & aqX(x,t)+bY(x,t)-\Phi\left(  t\right)
Y\left(  x,t\right)  .
\end{array}
\label{system2}%
\end{equation}
By taking $Z(x,t)=\frac{X(x,t)}{Y(x,t)}$, system (\ref{system2}) becomes%
\[
\frac{\partial Z(x,t)}{\partial t}=Z(x,t)\left\{  a\left(  1-q\right)
-aqZ(x,t)-b\right\}  .
\]
Assuming that concentration $Z(x,t)$ achieves a steady concentration
$\overline{Z}(x)$ over the time, we get%
\[
\overline{Z}(x)=\frac{a\left(  1-q\right)  -b}{aq}.
\]
The original population persists, i.e. the sequences in $\mathbb{A}$ survive
in a long term, if and only if $\overline{Z}(x)>0$, i.e. if and only if
\[
1-q>\frac{b}{a}.
\]
By writing $\frac{b}{a}=1-s$, with $s\in\left(  0,1\right)  $, \ the error
threshold \ is given by%
\begin{equation}
q<s. \label{error_thr}%
\end{equation}
This is exactly the classical condition determining the error threshold, see
e.g. \cite{Maynad Smith}, \cite{Szat-PTRSB}. Then, the error threshold
phenomenon occurs independently of the topology of the space of sequences.

\subsubsection{An example}

We take $(\mathbb{Y},\mu)=\left(  \mathbb{Z}_{p},dx\right)  $\ and consider
the mutation measures supported in the unit ball of the form
\begin{equation}
Q_{0}\left(  \left\vert x\right\vert _{p};\sigma,\alpha\right)  =Q_{0}\left(
\left\vert x\right\vert _{p}\right)  =\mathcal{N}\Omega\left(  \left\vert
x\right\vert _{p}\right)  \exp(-\sigma\left\vert x\right\vert _{p}^{\alpha}),
\label{EC_1}%
\end{equation}
where for $\sigma,\alpha>0$, and
\begin{equation}
\mathcal{N}\left(  \sigma,\alpha\right)  =%
{\displaystyle\int\limits_{\mathbb{Z}_{p}}}
\exp(-\sigma\left\vert x\right\vert _{p}^{\alpha})\text{ }dx. \label{EC_2}%
\end{equation}

Then, $Q_{0}\left(  \left\vert x\right\vert _{p}\right)  dx$ gives rise to a
family of mutation measures. We now fix a sequence $0\in\mathbb{Z}_{p}$, which
plays the role of the master sequence, and divide the space of sequences
$\mathbb{Z}_{p}$ into two subsets: $p^{M}\mathbb{Z}_{p}$ and $\mathbb{Z}%
_{p}\smallsetminus p^{M}\mathbb{Z}_{p}$ for some positive integer $M$. The set
$p^{M}\mathbb{Z}_{p}$ consists of the sequences in the unit ball that coincide
with the sequence $0$ up to the digit $M-1$. We also assume that%
\[
f\mid_{p^{M}\mathbb{Z}_{p}}\equiv a\text{, }f\mid_{\mathbb{Z}_{p}%
\smallsetminus p^{M}\mathbb{Z}_{p}}\equiv b\text{, with }a>b\text{.}%
\]
The probability $q\left(  \sigma,\alpha\right)  $ that a sequence in the set
$p^{M}\mathbb{Z}_{p}$ mutates into a sequence belonging to the set
$\mathbb{Z}_{p}\smallsetminus p^{M}\mathbb{Z}_{p}$ satisfies
\begin{gather*}
q\left(  \sigma,\alpha\right)  =%
{\displaystyle\int\limits_{p^{M}\mathbb{Z}_{p}}}
\text{ }%
{\displaystyle\int\limits_{\mathbb{Z}_{p}\smallsetminus p^{M}\mathbb{Z}_{p}}}
Q_{0}\left(  \left\vert x-y\right\vert _{p}\right)  dydx=%
{\displaystyle\int\limits_{p^{M}\mathbb{Z}_{p}}}
\text{ }%
{\displaystyle\int\limits_{\mathbb{Z}_{p}\smallsetminus p^{M}\mathbb{Z}_{p}}}
Q_{0}\left(  \left\vert y\right\vert _{p}\right)  dydx\\
=p^{-M}\text{ }%
{\displaystyle\int\limits_{\mathbb{Z}_{p}\smallsetminus p^{M}\mathbb{Z}_{p}}}
Q_{0}\left(  \left\vert y\right\vert _{p}\right)  dy>p^{-M}\text{ }%
{\displaystyle\int\limits_{\left\vert y\right\vert _{p}=p^{-M+1}}}
Q_{0}\left(  \left\vert y\right\vert _{p}\right)  dy\\
=\left(  1-p^{-1}\right)  p^{-2M+1}\mathcal{N}\exp(-\sigma p^{\left(
-M+1\right)  \alpha})\geq\left(  p-1\right)  p^{-2M}\mathcal{N}\exp\left(
-\sigma\right)  ,
\end{gather*}
where we used that $\exp(-\sigma p^{\left(  -M+1\right)  \alpha})\geq
\exp\left(  -\sigma\right)  $ for $M\geq1$.

We analyze now wether or not the condition (\ref{error_thr}) is satisfied,
when $M$ is fixed. The condition $M$ fixed can be relaxed to `$M$ is upper
bounded.' Taking into account that $s>0$ can be arbitrarily close to zero,
then there exists $M_{c}\geq M$ such that
\[
q\left(  \alpha\right)  >\left(  p-1\right)  p^{-2M_{c}}\mathcal{N}\exp\left(
-\sigma\right)  \geq s,
\]
which implies the existence of a classical error threshold:%
\begin{equation}
M_{c}\leq-\frac{\ln s}{2\ln p}-\frac{\sigma}{2\ln p}+\frac{\ln\left(
p-1\right)  \mathcal{N}}{2\ln p}\text{ for }s\in\left(  0,1\right)
\text{.}\label{EC_5}%
\end{equation}

If $M$ can grow, the condition (\ref{error_thr}) is satisfied if $\ \left(
p-1\right)  p^{-2M_{c}}\mathcal{N}\exp\left(  -\sigma\right)  <q\left(
\alpha\right)  <s$, which implies that%
\[
M>-\frac{\ln s}{2\ln p}-\frac{\sigma}{2\ln p}+\frac{\ln\left(  p-1\right)
\mathcal{N}}{2\ln p}\text{ for }s\in\left(  0,1\right)  \text{.}%
\]
Under a `fierce competition'\ between the groups $p^{M}\mathbb{Z}_{p}$,
$\mathbb{Z}_{p}\smallsetminus p^{M}\mathbb{Z}_{p}$, i.e. when rate $b$
approaches from the left to rate $a$ (i.e. $s\rightarrow0^{+}$), $M$ must
grow, which means that the survival of the sequences in the group
$p^{M}\mathbb{Z}_{p}$ demands that they get closer to master sequence $0$,
which means, that they must increase their lengths. Then, in this model \ the
`classical Eigen's paradox does not occur' because the length of the genomes
can grow during the evolution process.

\section{\label{Section_4}The $p$-adic Eigen-Model in the unit ball with a
radial mutation measure}

In this section we show the existence of a solution for the Cauchy problem
associated with (\ref{eq1})-(\ref{eq1A}). This goal is achieved by using the
classical method of separation of variables and $p$-adic wavelets, several
preliminary results are required.

\subsection{$p$-adic wavelets and pseudo-differential operators}

We take $\mathbb{K}=\mathbb{C}$, $\mathbb{R}$. We denote by $C(\mathbb{Q}%
_{p},\mathbb{K})$ the $\mathbb{K}$-vector space of \ continuous $\mathbb{K}%
$-valued functions defined on $\mathbb{Q}_{p}$.

We fix a function $\mathfrak{a}:\mathbb{R}_{+}\rightarrow\mathbb{R}_{+}$ and
define the pseudo-differential operator%
\[%
\begin{array}
[c]{ccc}%
\mathcal{D} & \rightarrow & C(\mathbb{Q}_{p},\mathbb{C})\cap L^{2}\\
&  & \\
\varphi & \rightarrow & \boldsymbol{A}\varphi,
\end{array}
\]
where $\left(  \boldsymbol{A}\varphi\right)  \left(  x\right)  =\mathcal{F}%
_{\xi\rightarrow x}^{-1}\left\{  \mathfrak{a}\left(  \left\vert \xi\right\vert
_{p}\right)  \mathcal{F}_{x\rightarrow\xi}\varphi\right\}  $.

The set of functions $\left\{  \Psi_{rnj}\right\}  $ defined as%
\begin{equation}
\Psi_{rnj}\left(  x\right)  =p^{\frac{-r}{2}}\chi_{p}\left(  p^{-1}j\left(
p^{r}x-n\right)  \right)  \Omega\left(  \left\vert p^{r}x-n\right\vert
_{p}\right)  , \label{eq4}%
\end{equation}
where $r\in\mathbb{Z}$, $j\in\left\{  1,\cdots,p-1\right\}  $, and $n$ runs
through a fixed set of representatives of $\mathbb{Q}_{p}/\mathbb{Z}_{p}$, is
an orthonormal basis of $L^{2}(\mathbb{Q}_{p})$ consisting of eigenvectors of
operator $\boldsymbol{A}$:%
\begin{equation}
\boldsymbol{A}\Psi_{rnj}=\mathfrak{a}(p^{1-r})\Psi_{rnj}\text{ for any
}r\text{, }n\text{, }j\text{,} \label{eq5}%
\end{equation}
see e.g. \cite[Theorem 3.29]{KKZuniga}, \cite[Theorem 9.4.2]{Alberio et
al}.\ Notice that%
\[
\widehat{\Psi}_{rnj}\left(  \xi\right)  =p^{\frac{r}{2}}\chi_{p}\left(
p^{-r}n\xi\right)  \Omega\left(  \left\vert p^{-r}\xi+p^{-1}j\right\vert
_{p}\right)  ,
\]
and then%
\[
\mathfrak{a}\left(  \left\vert \xi\right\vert _{p}\right)  \widehat{\Psi
}_{rnj}\left(  \xi\right)  =\mathfrak{a}(p^{1-r})\widehat{\Psi}_{rnj}\left(
\xi\right)  .
\]

\begin{remark}
From now on, we take $Q_{0}\left(  \left\vert x\right\vert _{p}\right)  $ to
be a real-valued, non-negative, radial function supported in $\mathbb{Z}_{p}$
satisfying $Q_{0}\in L^{1}\left(  \mathbb{Z}_{p}\right)  \cap L^{2}\left(
\mathbb{Z}_{p}\right)  $, and $\left\Vert Q_{0}\right\Vert _{1}=1$. By
extending $Q_{0}$ as zero out of $\mathbb{Z}_{p}$, we assume that $Q_{0}\in
L^{1}\left(  \mathbb{Q}_{p}\right)  \cap L^{2}\left(  \mathbb{Q}_{p}\right)
$. The Fourier transform $\widehat{Q_{0}}$ of $Q_{0}$ is a real-valued,
continuous function, which is radial in $\mathbb{Q}_{p}\smallsetminus\left\{
0\right\}  $, \ satisfying $\widehat{Q_{0}}\left(  0\right)  =1$, for this
reason, we use the notation $\widehat{Q}_{0}\left(  \left\vert \xi\right\vert
_{p}\right)  $.
\end{remark}

We now define
\[
\boldsymbol{B}\varphi=Q_{0}\ast\varphi\text{ for }\varphi\in\mathcal{D}\left(
\mathbb{Z}_{p}\right)  .
\]
Notice that the support of $Q_{0}\ast\varphi$ is $\mathbb{Z}_{p}$ since it is
an additive group. Since $\boldsymbol{B}\varphi\left(  x\right)
=\mathcal{F}_{\xi\rightarrow x}^{-1}(\widehat{Q_{0}}\left(  \left\vert
\xi\right\vert _{p}\right)  \mathcal{F}_{x\rightarrow\xi}\varphi)$, we have
\begin{equation}
\boldsymbol{B}\Psi_{rnj}\left(  x\right)  =\widehat{Q_{0}}\left(
p^{1-r}\right)  \Psi_{rnj}\left(  x\right)  , \label{nota_op_W}%
\end{equation}
where $\widehat{Q_{0}}\left(  p^{1-r}\right)  $ is a real number satisfying
$\left\vert \widehat{Q}_{0}\left(  p^{1-r}\right)  \right\vert \leq1$, and
$\Psi_{rnj}\left(  x\right)  $ is supported in the unit ball.

\subsection{$p$-adic wavelets supported in balls}

Notice that the restriction of $\Psi_{rnj}\left(  x\right)  $ to the ball
$I+p^{R_{0}}\mathbb{Z}_{p}$ has the form%
\begin{gather}
\Omega\left(  p^{R_{0}}\left\vert x-I\right\vert _{p}\right)  \Psi
_{rnj}\left(  x\right)  =\label{Table}\\
\left\{
\begin{array}
[c]{ll}%
\Psi_{rnj}\left(  x\right)  & \text{if }np^{-r}-I\in p^{R_{0}}\mathbb{Z}%
_{p}\text{, }r\leq-R_{0}\\
& \\
p^{\frac{-r}{2}}\Omega\left(  p^{R_{0}}\left\vert x-I\right\vert _{p}\right)
& \text{if }np^{-r}-I\in p^{-r}\mathbb{Z}_{p}\text{, }r\geq-R_{0}+1\\
& \\
0 & \text{if }np^{-r}-I\notin p^{-r}\mathbb{Z}_{p}\text{, }r\geq-R_{0}+1.
\end{array}
\right. \nonumber
\end{gather}

\begin{remark}
\label{Nota_Wavelets}With the above notation,%
\begin{gather*}
\left\{  \Psi_{rnj}\left(  x\right)  ;\text{supp }\Psi_{rnj}\left(  x\right)
\subseteq I+p^{R_{0}}\mathbb{Z}_{p}\right\}  =\\
\left\{  \Psi_{rnj}\left(  x+I\right)  ;\text{supp }\Psi_{rnj}\left(
x\right)  \subseteq p^{R_{0}}\mathbb{Z}_{p}\right\}  .
\end{gather*}
This observation is a very particular case of a general result asserting that
the basis (\ref{eq4}) is the orbit under a group of matrices of a mother
wavelet, see \cite[Theorem 9]{Albeverio-Kozyrev}.
\end{remark}

\begin{proposition}
The set of functions%
\begin{equation}
\left\{  \Omega\left(  p^{R_{0}}\left\vert x\right\vert _{p}\right)  \right\}
\bigcup\bigcup\limits_{j\in\left\{  1,\ldots,p-1\right\}  }\text{ }%
\bigcup\limits_{r\leq-R_{0}}\text{ }\bigcup\limits_{\substack{np^{-r}\in
p^{R_{0}}\mathbb{Z}_{p}\\n\in\mathbb{Q}_{p}/\mathbb{Z}_{p}}}\left\{
\Psi_{rnj}\left(  x\right)  \right\}  \label{Basis}%
\end{equation}
is an orthonormal basis of $L^{2}\left(  p^{R_{0}}\mathbb{Z}_{p}\right)  $.
\end{proposition}

\begin{proof}
In the demonstration we use the following results:

\textbf{Lemma A} (see e.g. \cite[Lemma 2.3.3]{Alberio et al}). Consider the
compact additive group $\left(  p^{R_{0}}\mathbb{Z}_{p},+\right)  $. Then any
nontrivial continuous additive character $\chi:p^{R_{0}}\mathbb{Z}%
_{p}\rightarrow S$ has the form $\chi\left(  x\right)  =\chi_{p}\left(
p^{-l}x\right)  $ for some positive integer $l\geq R_{0}$.

\textbf{Lemma B} (see e.g. \cite[Proposition 7.2.2]{Igusa}) Consider the
pre-Hilbert space $\left(  \mathcal{D}(p^{R_{0}}\mathbb{Z}_{p}),\left\langle
\cdot,\cdot\right\rangle \right)  $, where $\left\langle \cdot,\cdot
\right\rangle $ denotes the standard inner product in $L^{2}(p^{R_{0}%
}\mathbb{Z}_{p})$. Let $\Gamma\left(  p^{R_{0}}\mathbb{Z}_{p}\right)  $ be the
group of continuous characters of \ $\left(  p^{R_{0}}\mathbb{Z}_{p},+\right)
$. Then $\Gamma\left(  p^{R_{0}}\mathbb{Z}_{p}\right)  $ forms an orthonormal
basis of $\mathcal{D}(p^{R_{0}}\mathbb{Z}_{p})$. More precisely, every
$\varphi\in\mathcal{D}(p^{R_{0}}\mathbb{Z}_{p})$ can be expressed as a finite
sum of the form%
\[
\varphi\left(  x\right)  =\sum\limits_{\chi\in\Gamma\left(  p^{R_{0}%
}\mathbb{Z}_{p}\right)  }c_{\chi}\chi\left(  x\right)  \text{, where }c_{\chi
}=\left\langle \varphi,\chi\right\rangle .
\]

We identify a character \ $\chi\in\Gamma\left(  p^{R_{0}}\mathbb{Z}%
_{p}\right)  $ with the function $\Omega\left(  p^{R_{0}}\left\vert
x\right\vert _{p}\right)  \chi\in L^{2}(p^{R_{0}}\mathbb{Z}_{p})$. We denote
by $Span\left(  \Gamma\left(  p^{R_{0}}\mathbb{Z}_{p}\right)  \right)  $ the
$\mathbb{C}$-vector space generated by the elements of $\Gamma\left(
p^{R_{0}}\mathbb{Z}_{p}\right)  $. \ We first show that $\Gamma\left(
p^{R_{0}}\mathbb{Z}_{p}\right)  $ is an orthonormal basis of $L^{2}(p^{R_{0}%
}\mathbb{Z}_{p})$, i.e. that
\begin{equation}
\overline{Span\left(  \Gamma\left(  p^{R_{0}}\mathbb{Z}_{p}\right)  \right)
}=L^{2}(p^{R_{0}}\mathbb{Z}_{p}), \label{assertion}%
\end{equation}
where the bar means the topological closure with respect to $\left\Vert
\cdot\right\Vert _{2}$. Given any $f\in L^{2}(p^{R_{0}}\mathbb{Z}_{p})$ and
any $\epsilon>0$, by using the fact that $\mathcal{D}(p^{R_{0}}\mathbb{Z}%
_{p})$ is dense in $L^{2}(p^{R_{0}}\mathbb{Z}_{p})$ see e.g. \cite[Proposition
4.3.3]{Alberio et al}, there is $\varphi\in\mathcal{D}(p^{R_{0}}\mathbb{Z}%
_{p})$ such that $\left\Vert f-\varphi\right\Vert _{2}<\epsilon$. Now by using
Lemma B, $\varphi\left(  x\right)  =\sum_{\chi\in\Gamma\left(  p^{R_{0}%
}\mathbb{Z}_{p}\right)  }c_{\chi}\chi\left(  x\right)  $,
\[
\left\Vert f-\sum_{\chi\in\Gamma\left(  p^{R_{0}}\mathbb{Z}_{p}\right)
}c_{\chi}\chi\left(  x\right)  \right\Vert _{2}<\epsilon\text{ with }%
\sum_{\chi\in\Gamma\left(  p^{R_{0}}\mathbb{Z}_{p}\right)  }c_{\chi}%
\chi\left(  x\right)  \in Span\left(  \Gamma\left(  p^{R_{0}}\mathbb{Z}%
_{p}\right)  \right)  \text{,}%
\]
which implies (\ref{assertion}).

Finally, to show that (\ref{Basis}) is an orthonormal basis of $L^{2}\left(
p^{R_{0}}\mathbb{Z}_{p}\right)  $, by Lemma A, it is sufficient to show that
each $\chi\in\Gamma\left(  p^{R_{0}}\mathbb{Z}_{p}\right)  $ can be
represented as a linear combination of elements of the set (\ref{Basis}).
Since the trivial character is exactly $\Omega\left(  p^{R_{0}}\left\vert
x\right\vert _{p}\right)  $, it is sufficient to show that a character
$\Omega\left(  p^{R_{0}}\left\vert x\right\vert _{p}\right)  \chi\left(
p^{-l}x\right)  $, $l\geq R_{0}+1$, $\ $is a linear combinations of wavelets
of the form $\Psi_{rnj}$, with $np^{-r}\in p^{R_{0}}\mathbb{Z}_{p}$,
$r\leq-R_{0}$. Since $\Omega\left(  p^{R_{0}}\left\vert x\right\vert
_{p}\right)  \chi\left(  p^{-l}x\right)  \in L^{2}(\mathbb{Q}_{p})$, we may
compute the Fourier series with respect to the orthonormal basis $\left\{
\Psi_{rnj}\right\}  _{rnj}$. By using Table \ref{Table} and the fact that
$\int_{p^{R_{0}}\mathbb{Z}_{p}}\chi\left(  p^{-l}x\right)  dx=0$, for $l\geq
R_{0}+1$, we conclude that the non-zero Fourier coefficients $C_{rnj}$ are
given by%
\begin{multline*}
C_{rnj}=\\
p^{\frac{-r}{2}}\int\limits_{p^{R_{0}}\mathbb{Z}_{p}}\overline{\chi\left(
p^{-l}x\right)  }\chi_{p}\left(  p^{-1}j\left(  p^{r}x-n\right)  \right)
\Omega\left(  \left\vert p^{r}x-n\right\vert _{p}\right)  dx,
\end{multline*}
for $np^{-r}\in p^{R_{0}}\mathbb{Z}_{p}$, $r\leq-R_{0}$, i.e.
\begin{align*}
C_{rnj}  &  =\\
&  p^{\frac{-r}{2}}\chi_{p}\left(  -p^{-1}jn\right)  \int\limits_{p^{-r}%
n+p^{-r}\mathbb{Z}_{p}}\chi_{p}\left(  x\left(  p^{r-1}j-p^{-l}\right)
\right)  dx\text{ }\\
&  =p^{\frac{r}{2}}\chi_{p}\left(  -p^{-r-l}n\right)  \int\limits_{\mathbb{Z}%
_{p}}\chi_{p}\left(  z\left(  p^{-1}j-p^{-r-l}\right)  \right)  dy\text{
\ (taking }x=p^{-r}n+p^{-r}z\text{)}\\
&  =p^{\frac{r}{2}}\chi_{p}\left(  -p^{-r-l}n\right)  \int\limits_{\mathbb{Z}%
_{p}}\chi_{p}\left(  z\left(  p^{-1}j\right)  \right)  dy\text{ if }-r\geq l\\
&  =0\text{ if }-r\geq l.
\end{align*}
Which implies that $\Omega\left(  p^{R_{0}}\left\vert x\right\vert
_{p}\right)  \chi\left(  p^{-l}x\right)  $, $l\geq R_{0}+1$, $\ $is a linear
combinations of wavelets of the form $\Psi_{rnj}$, with $np^{-r}\in p^{R_{0}%
}\mathbb{Z}_{p}$, $r\leq-R_{0}$.
\end{proof}

\begin{remark}
Let $I\in\mathbb{Q}_{p}\smallsetminus p^{R_{0}}\mathbb{Z}_{p}$. By using the
isometry%
\[%
\begin{array}
[c]{ccc}%
L^{2}(I+p^{R_{0}}\mathbb{Z}_{p}) & \rightarrow & L^{2}(p^{R_{0}}\mathbb{Z}%
_{p})\\
&  & \\
f(x) & \rightarrow & f(x+I),
\end{array}
\]
we have that any $f\in L^{2}(I+p^{R_{0}}\mathbb{Z}_{p})$ admits a Fourier
expansion of the form%
\[
f(x)=C_{0}\Omega\left(  p^{R_{0}}\left\vert x-I\right\vert _{p}\right)
+\sum\limits_{rnj}C_{rnj}\Psi_{rnj}\left(  x-I\right)  ,\text{ }%
\]
where $x\in I+p^{R_{0}}\mathbb{Z}_{p}$, $j\in\left\{  1,\ldots,p-1\right\}  $,
$n\in\mathbb{Q}_{p}/\mathbb{Z}_{p}$, $r\leq-R_{0}$, $np^{-r}\in p^{R_{0}%
}\mathbb{Z}_{p}$. By Remark \ref{Nota_Wavelets},%
\begin{multline*}
f(x)=C_{0}\Omega\left(  p^{R_{0}}\left\vert x-I\right\vert _{p}\right)  +\\
\sum\limits_{rnj\text{, supp}\Psi_{rnj}\left(  x\right)  \subseteq I+p^{R_{0}%
}\mathbb{Z}_{p}\text{ }}C_{rnj}\Psi_{rnj}\left(  x\right)  ,
\end{multline*}
for $x\in I+p^{R_{0}}\mathbb{Z}_{p}$. We now set
\[
L_{0}^{2}(I,R_{0})=\left\{  f\in L^{2}(I+p^{R_{0}}\mathbb{Z}_{p});%
{\displaystyle\int\limits_{I+p^{R_{0}}\mathbb{Z}_{p}}}
f\text{ }dx=0\right\}  .
\]
In conclusion, we have the following result:
\end{remark}

\begin{proposition}
\label{Theorem_1}The space $L^{2}(I+p^{R_{0}}\mathbb{Z}_{p})$ satisfies
\[
L^{2}(I+p^{R_{0}}\mathbb{Z}_{p})=\mathbb{C}\Omega\left(  p^{R_{0}}\left\vert
x-I\right\vert _{p}\right)
{\displaystyle\bigoplus}
L_{0}^{2}(I,R_{0})
\]
and
\begin{align}
L^{2}(I+p^{R_{0}}\mathbb{Z}_{p})  &  =%
{\displaystyle\bigoplus\limits_{j\in\left\{  1,\ldots,p-1\right\}  }}
\text{ }%
{\displaystyle\bigoplus\limits_{r\leq-R_{0}}}
\text{ }%
{\displaystyle\bigoplus\limits_{\substack{n\in\mathbb{Q}_{p}/\mathbb{Z}%
_{p}\\np^{-r}\in p^{R_{0}}\mathbb{Z}_{p}}}}
\overline{Span\left(  \Psi_{rnj}\left(  x-I\right)  \right)  }\nonumber\\
&  =\text{ }%
{\displaystyle\bigoplus\limits_{rnj\text{, supp}\Psi_{rnj}\left(  x\right)
\subseteq I+p^{R_{0}}\mathbb{Z}_{p}}}
\overline{Span\left(  \Psi_{rnj}(x\right)  )}. \label{Definition_L_0}%
\end{align}

\end{proposition}

The construction of orthonormal basis for $L^{2}(\mathbb{Z}_{p})$ has been
widely considered in the literature, see e.g. \cite[Theorem 4]{Bikulov},
\cite[Theore 1, Proposition 1, 2]{Kosyak et al} and the references therein.
However, we have not specifically found Proposition \ref{Theorem_1} in the literature.

\subsection{The operator $\boldsymbol{W}_{0}$}

From now on, we assume that the fitness function $f$ is a non-negative test
function supported in $\mathbb{Z}_{p}$, independent of the time, of the form%
\[
f(x)=%
{\displaystyle\sum\limits_{I\in G_{M}}}
f(I)\Omega\left(  p^{M}\left\vert x-I\right\vert _{p}\right)  ,
\]
where $M$ is fixed positive integer, $G_{M}=\mathbb{Z}_{p}/p^{M}\mathbb{Z}%
_{p}$, and with $f(I)>0$ for $I\in G_{M}$. For $\varphi\in L^{2}\left(
\mathbb{Z}_{p}\right)  $, we define the operator%
\[
\boldsymbol{W}_{0}\varphi\left(  x\right)  =Q_{0}\left(  \left\vert
x\right\vert _{p}\right)  \ast\left(  f(x\right)  \varphi\left(  x\right)  ).
\]
Then $\boldsymbol{W}_{0}:L^{2}\left(  \mathbb{Z}_{p}\right)  \rightarrow
L^{2}\left(  \mathbb{Z}_{p}\right)  $ is a linear bounded operator.

We set $\mathcal{D}_{M}$ for the $\mathbb{R}$-vector space generated by
$\left\{  \Omega\left(  p^{M}\left\vert x-I\right\vert _{p}\right)  \right\}
_{I\in G_{M}}$ as before. The space $\mathcal{D}_{M}$ is invariant under
$\boldsymbol{W}_{0}$. Indeed,%
\begin{align*}
\boldsymbol{W}_{0}\Omega\left(  p^{M}\left\vert x-I\right\vert _{p}\right)
&  =f(I)%
{\displaystyle\int\limits_{I+p^{M}\mathbb{Z}_{p}}}
Q_{0}\left(  \left\vert x-y\right\vert _{p}\right)  dy=f(I)%
{\displaystyle\int\limits_{x-I+p^{M}\mathbb{Z}_{p}}}
Q_{0}\left(  \left\vert z\right\vert _{p}\right)  dz\\
&  =\left\{
\begin{array}
[c]{ccc}%
f(I)%
{\displaystyle\int\limits_{p^{M}\mathbb{Z}_{p}}}
Q_{0}\left(  \left\vert z\right\vert _{p}\right)  dz & \text{if} & x\in
I+p^{M}\mathbb{Z}_{p}\\
&  & \\
f(I)Q_{0}\left(  \left\vert J-I\right\vert _{p}\right)  p^{-M} & \text{if} &
x\in J+p^{M}\mathbb{Z}_{p},\text{ }I\neq J,
\end{array}
\right.
\end{align*}
and thus%
\begin{multline*}
\boldsymbol{W}_{0}\Omega\left(  p^{M}\left\vert x-I\right\vert _{p}\right)  =%
{\displaystyle\sum\limits_{\substack{J\in G_{M}\\J\neq I}}}
Q_{0}\left(  \left\vert J-I\right\vert _{p}\right)  p^{-M}\Omega\left(
p^{M}\left\vert x-J\right\vert _{p}\right) \\
+\left(  \text{ }%
{\displaystyle\int\limits_{p^{M}\mathbb{Z}_{p}}}
Q_{0}\left(  \left\vert z\right\vert _{p}\right)  dz\right)  \Omega\left(
p^{M}\left\vert x-I\right\vert _{p}\right)  .
\end{multline*}
We set $\varphi_{I}\left(  x\right)  :=\Omega\left(  p^{M}\left\vert
x-I\right\vert _{p}\right)  $, and denote by $\left[  \varphi_{I}\right]
_{I\in G_{M},}$ a column vector, then on $\mathcal{D}_{M}$ operator
$\boldsymbol{W}_{0}$ is represented by the matrix $\mathbb{W}^{0}:=\left[
\mathbb{W}_{I,J}^{0}\right]  _{I,J\in G_{M}}$, \ with
\begin{equation}
\text{ }\mathbb{W}_{I,J}^{0}=\left\{
\begin{array}
[c]{ccc}%
f(I)%
{\displaystyle\int\limits_{p^{M}\mathbb{Z}_{p}}}
Q_{0}\left(  \left\vert z\right\vert _{p}\right)  dz & \text{if} & I=J\\
&  & \\
f(I)Q_{0}\left(  \left\vert J-I\right\vert _{p}\right)  p^{-M} &  & I\neq J,
\end{array}
\right.  \label{Matrix_W_0}%
\end{equation}
Now the space $L_{0}^{2}(I,M)$, see (\ref{Definition_L_0}), is invariant under
$\boldsymbol{W}_{0}$ since $\boldsymbol{W}_{0}\Psi_{rnj}\left(  x\right)
=\widehat{Q_{0}}\left(  p^{1-r}\right)  \Psi_{rnj}\left(  x\right)  $ for
$\Psi_{rnj}\left(  x\right)  \in L_{0}^{2}(I,M)$. We now attach to
$\boldsymbol{W}_{0}$ the real vector space%
\[
L\left(  \boldsymbol{W}_{0}\right)  :=\mathcal{D}_{M}%
{\displaystyle\bigoplus}
{\displaystyle\bigoplus\limits_{I\in G_{M}}}
\left\{  L_{0}^{2}(I,M)\cap L_{\mathbb{R}}^{2}(I+p^{M}\mathbb{Z}_{p})\right\}
.
\]

\begin{lemma}
\label{Lemma_1}The space $L\left(  \boldsymbol{W}_{0}\right)  $ is invariant
under operator $\boldsymbol{W}_{0}$, i.e. $\boldsymbol{W}_{0}\left(  L\left(
\boldsymbol{W}_{0}\right)  \right)  \subset L\left(  \boldsymbol{W}%
_{0}\right)  $, and $\boldsymbol{W}_{0}\left(  L_{\mathbb{R}}^{2}%
(\mathbb{Z}_{p})\right)  \subset L\left(  \boldsymbol{W}_{0}\right)  $.
\end{lemma}

\subsection{The Cauchy problem for operator $\boldsymbol{W}_{0}$}

We now consider the following initial value problem:%
\begin{equation}
\left\{
\begin{array}
[c]{ll}%
Y:\mathbb{Q}_{p}\times\mathbb{R}_{+}\rightarrow\mathbb{R}\text{,} &
\begin{array}
[c]{l}%
Y\left(  \cdot,t\right)  \in L\left(  \boldsymbol{W}_{0}\right)  \cap
L_{\mathbb{R}}^{2}(\mathbb{Z}_{p})\text{, }\\
Y\left(  x,\cdot\right)  \in C^{1}\left(  \mathbb{R}_{+},\mathbb{R}\right)
\end{array}
\\
& \\
\frac{dY\left(  x,t\right)  }{dt}=\boldsymbol{W}_{0}Y\left(  x,t\right)  , &
x\in\mathbb{Q}_{p},t>0\\
& \\
Y\left(  x,0\right)  =Y_{0}\left(  x\right)  \in L\left(  \boldsymbol{W}%
_{0}\right)  \cap L_{\mathbb{R}}^{2}(\mathbb{Z}_{p}). &
\end{array}
\right.  \label{eq6}%
\end{equation}

We solve (\ref{eq6}) by using the separation of variables method. We first
look for a complex-valued solution of (\ref{eq6}) of the form%
\begin{equation}
\widetilde{Y}\left(  x,t\right)  =%
{\displaystyle\sum\limits_{I\in G_{M}}}
C_{I}^{0}(t)\varphi_{I}\left(  x\right)  +%
{\displaystyle\sum\limits_{I\in G_{M}}}
\text{ }\sum\limits_{\text{supp}\Psi_{rnj}\subseteq I+p^{M}\mathbb{Z}_{p}%
}C_{rjn}^{I}\left(  t\right)  \Psi_{rnj}\left(  x\right)  ,\nonumber
\end{equation}
where $C_{rjn}^{I}\left(  t\right)  $ are complex-valued functions, which
admit continuous temporal derivatives. By replacing%
\begin{multline*}
\frac{d}{dt}\widetilde{Y}\left(  x,t\right)  =%
{\displaystyle\sum\limits_{I\in G_{M}}}
\left(  \frac{d}{dt}C_{I}^{0}(t)\right)  \varphi_{I}\left(  x\right) \\
+%
{\displaystyle\sum\limits_{I\in G_{M}}}
\text{ \ }\sum\limits_{\text{ supp}\Psi_{rnj}\subseteq I+p^{M}\mathbb{Z}_{p}%
}\left(  \frac{d}{dt}C_{rjn}^{I}\left(  t\right)  \right)  \Psi_{rnj}\left(
x\right)  ,
\end{multline*}
and%
\begin{multline*}
\boldsymbol{W}_{0}\widetilde{Y}\left(  x,t\right)  =\boldsymbol{W}_{0}\left(
{\displaystyle\sum\limits_{I\in G_{M}}}
C_{I}^{0}(t)\varphi_{I}\left(  x\right)  \right) \\
+%
{\displaystyle\sum\limits_{I\in G_{M}}}
\text{ \ }\sum\limits_{\text{supp}\Psi_{rnj}\subseteq I+p^{M}\mathbb{Z}_{p}%
}\widehat{Q_{0}}\left(  p^{1-r}\right)  f(I)C_{rjn}^{I}\left(  t\right)
\Psi_{rnj}\left(  x\right)  ,
\end{multline*}
in (\ref{eq6}), we get the following systems of differential equations:
\[
\frac{d}{dt}C_{rjn}^{I}\left(  t\right)  =\widehat{Q_{0}}\left(
p^{1-r}\right)  f(I)C_{rjn}^{I}\left(  t\right)  \text{ for }I\in G_{M},
\]%
\[
\frac{d}{dt}\left[  C_{I}^{0}(t)\right]  _{I\in G_{M}}=\mathbb{W}^{0}\left[
C_{I}^{0}(t)\right]  _{I\in G_{M}}.
\]
Therefore%
\begin{multline*}
\widetilde{Y}\left(  x,t\right)  =\left(  e^{t\mathbb{W}^{0}}\left[  C_{I}%
^{0}(0)\right]  _{I\in G_{M}}\right)  \left[  \varphi_{I}\left(  x\right)
\right]  _{I\in G_{M}}^{T}\\
+%
{\displaystyle\sum\limits_{I\in G_{M}}}
\text{ \ }\sum\limits_{\text{supp}\Psi_{rnj}\subseteq I+p^{M}\mathbb{Z}_{p}%
}e^{t\widehat{Q_{0}}\left(  p^{1-r}\right)  f(I)}C_{rjn}^{I}\left(  0\right)
\Psi_{rnj}\left(  x\right)  ,
\end{multline*}
where $\left[  \varphi_{I}\left(  x\right)  \right]  _{I\in G_{M}}^{T}$denotes
the transpose of the column vector $\left[  \varphi_{I}\left(  x\right)
\right]  _{I\in G_{M}}$. Notice that
\begin{equation}%
{\displaystyle\sum\limits_{I\in G_{M}}}
\text{ \ }\sum\limits_{\text{supp}\Psi_{rnj}\subseteq I+p^{M}\mathbb{Z}_{p}%
}e^{t\widehat{Q_{0}}\left(  p^{1-r}\right)  f(I)}C_{rjn}^{I}\left(  0\right)
\Psi_{rnj}\left(  x\right)  \in%
{\displaystyle\bigoplus\limits_{I\in G_{M}}}
L_{0}^{2}\left(  I+p^{M}\mathbb{Z}_{p}\right)  . \label{average_0}%
\end{equation}
The constants $C_{I}^{0}(0)$, $C_{rjn}^{I}\left(  0\right)  $ are determined
by the Fourier expansion of the initial datum $Y_{0}\left(  x\right)  $, which
is a real-valued function. Then $C_{I}^{0}(0)\in\mathbb{R}$ for any $I$, and
$C_{rjn}^{I}\left(  0\right)  \in\mathbb{C}$ for any $I$, $rjn$, and
\begin{multline*}
Y\left(  x,t\right)  =\operatorname{Re}\left(  \widetilde{Y}\left(
x,t\right)  \right)  =\left(  e^{t\mathbb{W}^{0}}\left[  C_{I}^{0}(0)\right]
_{I\in G_{M}}\right)  \left[  \varphi_{I}\left(  x\right)  \right]  _{I\in
G_{M}}^{T}\\
+%
{\displaystyle\sum\limits_{I\in G_{M}}}
\text{ \ }\sum\limits_{\text{supp}\Psi_{rnj}\subseteq I+p^{M}\mathbb{Z}_{p}%
}e^{t\widehat{Q_{0}}\left(  p^{1-r}\right)  f(I)}\operatorname{Re}\left\{
C_{rjn}^{I}\left(  0\right)  \Psi_{rnj}\left(  x\right)  \right\}  ,
\end{multline*}
and due to (\ref{average_0}),
\begin{align}
\overline{Y\left(  t\right)  }  &  :=%
{\displaystyle\int\limits_{\mathbb{Z}_{p}}}
Y\left(  x,t\right)  dx=%
{\displaystyle\int\limits_{\mathbb{Z}_{p}}}
\left(  e^{t\mathbb{W}^{0}}\left[  C_{I}^{0}(0)\right]  _{I\in G_{M}}\right)
\left[  \varphi_{I}\left(  x\right)  \right]  _{I\in G_{M}}^{T}%
dx\label{Average}\\
&  =p^{-M}\left(  e^{t\mathbb{W}^{0}}\left[  C_{I}^{0}(0)\right]  _{I\in
G_{M}}\right)  \boldsymbol{1}^{T},\nonumber
\end{align}
where $\boldsymbol{1}^{T}\boldsymbol{=}\left[  1,1,\ldots,1,1\right]  $.

\subsection{The Cauchy problem for the $p$-adic Eigen-Schuster equation in the
unit ball}

We assume that the fitness function and the mutation measure are supported in
the unit ball and that they are time independent. We now consider the $p$-adic
Eigen-Schuster equation in the unit ball:%
\begin{equation}
\frac{\partial X\left(  x,t\right)  }{\partial t}=\boldsymbol{W}_{0}X\left(
x,t\right)  -\Phi\left(  t\right)  X\left(  x,t\right)  \text{, }%
x\in\mathbb{Z}_{p}\text{, }t\in\mathbb{R}_{+}\text{,} \label{Eq_E_S_ball}%
\end{equation}
where%
\[
\Phi\left(  t,X\right)  :=\Phi\left(  t\right)  =%
{\displaystyle\int\limits_{\mathbb{Z}_{p}}}
f\left(  y\right)  X\left(  y,t\right)  dy\text{ \ \ for }t\geq0\text{.}%
\]
By changing variables as
\begin{align*}
X\left(  x,t\right)   &  =Y\left(  x,t\right)  \exp\left(  -\int_{0}^{t}%
\Phi\left(  \tau\right)  d\tau\right) \\
&  =\frac{Y\left(  x,t\right)  }{%
{\displaystyle\int\limits_{\mathbb{Z}_{p}}}
Y\left(  x,t\right)  dx}=:\frac{Y\left(  x,t\right)  }{\overline{Y\left(
t\right)  }}%
\end{align*}
(\ref{Eq_E_S_ball}) becomes%
\[
\frac{\partial Y\left(  x,t\right)  }{\partial t}=\boldsymbol{W}_{0}Y\left(
x,t\right)  ,\text{ , }x\in\mathbb{Z}_{p}\text{, }t\in\mathbb{R}_{+}\text{.}%
\]
Therefore%
\begin{gather}
X\left(  x,t\right)  =\frac{\left(  e^{t\mathbb{W}^{0}}\left[  C_{I}%
^{0}(0)\right]  _{I\in G_{M}}\right)  \left[  \varphi_{I}\left(  x\right)
\right]  _{I\in G_{M}}^{T}}{\overline{Y\left(  t\right)  }}
\label{Eq_solution}\\
+%
{\displaystyle\sum\limits_{I\in G_{M}}}
\text{ \ }\sum\limits_{\text{supp}\Psi_{rnj}\subseteq I+p^{M}\mathbb{Z}_{p}%
}\frac{\text{ }e^{t\widehat{Q_{0}}\left(  p^{1-r}\right)  f(I)}}%
{\overline{Y\left(  t\right)  }}\operatorname{Re}\left(  C_{rjn}^{I}\left(
0\right)  \Psi_{rnj}\left(  x\right)  \right)  ,\nonumber
\end{gather}
where $\overline{Y\left(  t\right)  }$\ is given in (\ref{Average}) and
\begin{multline*}
\operatorname{Re}\left(  C_{rjn}^{I}\left(  0\right)  \Psi_{rnj}\left(
x\right)  \right)  =p^{\frac{-r}{2}}\operatorname{Re}\left(  C_{rjn}%
^{I}\left(  0\right)  \right)  \cos\left(  p^{r-1}jx\right)  \Omega\left(
\left\vert p^{r}x-n\right\vert _{p}\right) \\
-p^{\frac{-r}{2}}\operatorname{Im}\left(  C_{rjn}^{I}\left(  0\right)
\right)  \sin\left(  p^{r-1}jx\right)  \Omega\left(  \left\vert p^{r}%
x-n\right\vert _{p}\right)  .
\end{multline*}

We now assume that $X_{0}(x):=X\left(  x,0\right)  $ is a real-valued function
supported in the unit ball of the form%
\begin{equation}
X_{0}(x)=%
{\displaystyle\sum\limits_{I\in G_{M}}}
X_{0}(I)\Omega\left(  p^{M}\left\vert x-I\right\vert _{p}\right)  ,
\label{Initial_datum}%
\end{equation}
with $G_{M}=\mathbb{Z}_{p}/p^{M}\mathbb{Z}_{p}$ as before and with
$X_{0}(I)>0$ for $I\in G_{M}$. Since
\[
\int_{I+p^{M}\mathbb{Z}_{p}}\Psi_{rnj}\left(  x\right)  =0
\]
for any $\Psi_{rnj}\left(  x\right)  $ with support contained in
$I+p^{M}\mathbb{Z}_{p}$, we get that $C_{rjn}^{I}\left(  0\right)  =0$ for any
$I$, $rnj$, in (\ref{Eq_solution}).

\begin{theorem}
\label{Theorem_0}The Cauchy problem (\ref{eq6}) admits a solution $X\left(
x,t\right)  $ of the form (\ref{Eq_solution}). Furthermore, if the initial
datum has the form (\ref{Initial_datum}), then
\[
X\left(  x,t\right)  =\frac{\left(  e^{t\mathbb{W}^{0}}\left[  C_{I}%
^{0}(0)\right]  _{I\in G_{M}}\right)  \left[  \varphi_{I}\left(  x\right)
\right]  _{I\in G_{M}}^{T}}{\overline{Y\left(  t\right)  }}.
\]
Which is a classical solution of the Eigen-Schuster model.
\end{theorem}

Since $\varphi_{I}\left(  x\right)  =\Omega\left(  p^{M}\left\vert
x-I\right\vert _{p}\right)  $, then the concentration $X\left(  x,t\right)  $
depends only on the $M$ first $p$-adic digits of $x$, and thus in this model
the length of the sequences do not change in time.

\section{\label{Section_5}The $p$-adic quasispecies in the unit ball}

Since $\mathbb{W}^{0}$ is a real symmetric matrix, it is diagonalizable and
all its eigenvalues are real. Let $\lambda_{\max}$ be the largest eigenvalue
of $\mathbb{W}^{0}$. Then%
\begin{gather*}
\lim_{t\rightarrow\infty}X\left(  x,t\right)  =\lim_{t\rightarrow\infty}%
\frac{\left(  e^{t\mathbb{W}^{0}}\left[  C_{I}^{0}(0)\right]  _{I\in G_{M}%
}\right)  \left[  \varphi_{I}\left(  x\right)  \right]  _{I\in G_{M}}^{T}%
}{\overline{Y\left(  t\right)  }}=\\
\lim_{t\rightarrow\infty}\frac{e^{t\lambda_{\max}}\sum_{I}p_{I}\left(
t\right)  \varphi_{I}\left(  x\right)  }{e^{t\lambda_{\max}}\sum_{I}%
p_{I}\left(  t\right)  }=%
{\textstyle\sum\limits_{j=1}^{M^{\prime}}}
c_{j}\varphi_{I_{j}}\left(  x\right)  ,
\end{gather*}
where the $p_{I}\left(  t\right)  $s are polynomials in $t$ and the $c_{j}$s
are positive constants. This situation corresponds to \textit{the survival of
the fitter}. This is the typical scenario predicted by the classical
Eigen-Schuster equation. In this context the Eigen paradox happens naturally.
Thus, to get a different asymptotic behavior of solution (\ref{Eq_solution}),
some of the oscillatory terms must be preserved in the\ long term. If the
oscillatory terms in (\ref{Eq_solution}) do not vanish in the long term, then
there are sequences (with arbitrary length) spread out throughout the unit
ball. This means that the Eigen paradox does not occur since in the long term
there are sequences of infinite length.

\begin{definition}
\label{Def_quasispecies}We say that a solution of the Cauchy problem
(\ref{eq6}) admits a quasispecies solution if there exist a solution
\ $X\left(  x,t\right)  $ of the form (\ref{Eq_solution}) satisfying that
$\int_{\mathbb{Q}_{p}}X\left(  y,t\right)  dy=1$ for $t>0$, and that the set
\[
\left\{  \left(  I,rnj\right)  ;\text{ }\lim_{t\rightarrow\infty}%
\frac{p^{\frac{-r}{2}}\text{ }e^{t\widehat{Q_{0}}\left(  p^{1-r}\right)
f(I)}\text{ }}{\overline{Y\left(  t\right)  }}\neq0\text{ when }C_{rjn}%
^{I}\left(  0\right)  \neq0\right\}
\]
is non-empty.
\end{definition}

\subsection{Some results about semigroups of matrices}

In order to establish the existence of the $p$-adic quasispecies we need
several preliminary results.

\subsubsection{Diagonally dominant matrices}

A real matrix $B=\left[  B_{ij}\right]  _{i,j\in I}$, $I=\left\{
1,\ldots,n\right\}  $, is said to be \textit{diagonally dominant}, if
\[
\left\vert B_{ii}\right\vert \geq%
{\displaystyle\sum\limits_{j\neq i}}
\left\vert B_{ij}\right\vert \text{ \ for any }i\in I.
\]
It is \textit{strictly} \textit{diagonally dominant} if
\[
\left\vert B_{ii}\right\vert >%
{\displaystyle\sum\limits_{j\neq i}}
\left\vert B_{ij}\right\vert \text{ \ for any }i\in I.
\]
Let $B=\left[  B_{ij}\right]  _{i,j\in I}$ be a strictly diagonally dominant.
Then (i) $B$ is non singular; (ii) if $B_{ii}>0$ for all $i\in I$, then every
eigenvalue of $B$ has a positive real part; (iii) if $B$ is symmetric and
$B_{ii}>0$ for all $i\in I$, then $B$ is positive definite, see \cite[Theorem
6.1.10]{Horn and Jhonson}.

We now apply this result to the matrix $\mathbb{W}^{0}=\left[  \mathbb{W}%
_{I,J}^{0}\right]  _{I,J\in G_{M}}$, see (\ref{Matrix_W_0}). Since%
\[%
{\displaystyle\int\limits_{\mathbb{Z}_{p}}}
Q_{0}\left(  \left\vert z\right\vert _{p}\right)  dz=%
{\displaystyle\int\limits_{p^{M}\mathbb{Z}_{p}}}
Q_{0}\left(  \left\vert z\right\vert _{p}\right)  dz+p^{-M}%
{\displaystyle\sum\limits_{J\neq I}}
Q_{0}\left(  \left\vert J-I\right\vert _{p}\right)  =1,
\]
we have%
\[
\mathbb{W}_{I,I}^{0}+%
{\displaystyle\sum\limits_{J\neq I}}
\mathbb{W}_{I,J}^{0}=f(I).
\]
We now introduce the hypothesis:%
\begin{equation}%
{\displaystyle\int\limits_{p^{M}\mathbb{Z}_{p}}}
Q_{0}\left(  \left\vert z\right\vert _{p}\right)  dz\in\left(  \frac{1}%
{2},1\right)  . \tag{Hypothesis A}%
\end{equation}
Under the hypothesis A,$\ \mathbb{W}_{I,I}^{0}>\frac{f(I)}{2}$, which implies
that $\mathbb{W}^{0}$ is a strictly diagonally dominant matrix. Then, we have
the following result:

\begin{lemma}
\label{Lemma_A}Under the Hypothesis A, the symmetric matrix $\mathbb{W}^{0}$
is\ a strictly diagonally dominant, nonsingular, and all its eigenvalues are positive.
\end{lemma}

\subsubsection{Semigroups of matrices}

Let $\left\{  \mu_{I}\right\}  _{I\in G_{M}}$ be the positive eigenvalues of
$\mathbb{W}^{0}$ repeated according their multiplicity. We set%
\[
\mu_{\max}:=\max_{_{I\in G_{M}}}\mu_{I}>0\text{.}%
\]
Then $\mathbb{W}^{0}=\mathbb{S}+\mathbb{O}$, where
\[
\mathbb{P}^{-1}\mathbb{SP}=diag\left[  \mu_{I}\right]  _{I\in G_{M}},
\]
for some inversible matrix $\mathbb{P}$, and $\mathbb{O}$ a nilpotent matrix
of order $k\leq\#G_{M}$, and%
\[
\boldsymbol{z}(t)=\mathbb{P}\left(  diag\left[  \mu_{I}\right]  _{I\in G_{M}%
}\right)  \mathbb{P}^{-1}\left\{  \mathbb{I}+\mathbb{O}t+\cdots+\frac
{\mathbb{O}^{k-1}t^{k-1}}{\left(  k-1\right)  !}\right\}  \boldsymbol{z}_{0}%
\]
is the solution $\frac{d}{dt}\boldsymbol{z}(t)=\mathbb{W}^{0}\boldsymbol{z}%
(t)$, $\boldsymbol{z}(0)=\boldsymbol{z}_{0}$, see e.g.\ \cite[Theorem 1,
Corollary 1]{Perko}. By applying this result to $\overline{Y\left(  t\right)
}$, see (\ref{Average}), we have the following result:

\begin{lemma}
\label{Lemma_B}$\overline{Y\left(  t\right)  }\leq Ct^{k-1}e^{\mu_{\max}t}$
for $t>0$.
\end{lemma}

\subsection{$p$-Adic quasispecies}

We now introduce the hypothesis:%
\begin{equation}
\widehat{Q_{0}}\left(  p^{1-r_{0}}\right)  f(I_{0})>\mu_{\max}\text{, }
\tag{Hypothesis B}%
\end{equation}
for some negative integer $r_{0}$ and $I_{0}\in G_{M}$. Which means that
operator $\boldsymbol{W}_{0}$ has a positive eigenvalue greater than
$\frac{\mu_{\max}}{f(I_{0})}$.

Under the Hypotheses A and B, by Lemma \ref{Lemma_B}, we have%
\begin{multline*}
\lim_{t\rightarrow\infty}\frac{p^{\frac{-r_{0}}{2}}\text{ }e^{t\widehat{Q_{0}%
}\left(  p^{1-r_{0}}\right)  f(I_{0})}\text{ }}{\overline{Y\left(  t\right)
}}\\
\geq\frac{p^{\frac{-r_{0}}{2}}}{C}\lim_{t\rightarrow\infty}t^{-k+1}%
e^{t\left\{  \widehat{Q_{0}}\left(  p^{1-r_{0}}\right)  f(I_{0})-\mu_{\max
}\right\}  }=\infty,
\end{multline*}
if $C_{r_{0}jn}^{I_{0}}\left(  0\right)  \neq0$ for some $n\in\mathbb{Q}%
_{p}/\mathbb{Z}_{p}$ satisfying $np^{-r_{0}}\in\mathbb{Z}_{p}$, see Table
\ref{Table}.

By (\ref{Eq_solution}), the initial condition $X(x,0)$ is completely
determined by a sequence from the set%
\begin{gather*}
\mathcal{S}:=\left\{  \left[  C_{I}^{0}(0)\right]  _{I\in G_{M}}\in
\mathbb{R}^{\#G_{M}};C_{I}^{0}(0)\neq0\text{ for some }I\in G_{M}\right\}
{\textstyle\bigsqcup}
\\%
{\textstyle\bigsqcup\limits_{I\in G_{M}}}
\left\{  C_{rjn}^{I}\left(  0\right)  \in\mathbb{C};r\leq0,j\in\left\{
1,\ldots,p-1\right\}  ,n\in\mathbb{Q}_{p}/\mathbb{Z}_{p}\text{ with }%
np^{-r}\in\mathbb{Z}_{p}\right\}  .
\end{gather*}
We use the notation $X(x,0)\in\mathcal{S}$ to mean that $X(x,0)$ is determined
by a sequence from $\mathcal{S}$. Notice that condition $C_{I}^{0}(0)\neq0$
for some $I\in G_{M}$ is needed to guarantee that $\int X(x,0)dx=1$. Take
$I_{0}$, $r_{0}$ such that Hypothesis B is satisfied, the condition
$C_{r_{0}jn}^{I_{0}}\left(  0\right)  \neq0$ defines a a subset $\mathcal{S}%
_{0}$ of $\mathcal{S}$.

\begin{theorem}
\label{Theorem_D}Under the Hypothesis A, B, and assuming that $X(x,0)\in
\mathcal{S}_{0}$, then the Cauchy problem (\ref{eq6}) admits a quasispecies solution.
\end{theorem}

\subsection{A family of Gibbs-type mutation measures}

In this section we present an infinite family of mutation measures supported
in the unit ball satisfying the Hypotheses A and B. More precisely,%
\begin{equation}
Q_{0}\left(  \left\vert x\right\vert _{p};\sigma,\alpha\right)  :=Q_{0}\left(
\left\vert x\right\vert _{p}\right)  =\mathcal{N}\Omega\left(  \left\vert
x\right\vert _{p}\right)  \exp(-\sigma\left\vert x\right\vert _{p}^{\alpha}),
\label{Mutation_Measure}%
\end{equation}
where $\sigma,\alpha>0$ , $\mathcal{N}=\int_{\mathbb{Z}_{p}}\exp
(-\sigma\left\vert x\right\vert _{p}^{\alpha})dx$, and $\Omega\left(
\left\vert x\right\vert _{p}\right)  $ is the characteristic function of the
unit ball. Notice that $Q_{0}\left(  \left\vert x\right\vert _{p}\right)  $ is
integrable. We set%
\[
Z(\xi;\sigma,\alpha):=%
{\displaystyle\int\limits_{\mathbb{Q}_{p}}}
\chi_{p}\left(  \xi x\right)  \exp(-\sigma\left\vert x\right\vert _{p}%
^{\alpha})dx,
\]
for $\sigma,\alpha>0$. This is the $p$-adic heat kernel widely studied in
connection with the $p$-adic heat equation, see e.g. \cite{Kochubei},
\cite{KKZuniga}, \cite{V-V-Z}, \cite{Zuniga-LNM-2016} The heat kernel
$Z(\xi;\sigma,\alpha)$ is non-negative, continuous function in $\xi$ for any
$\sigma,\alpha>0$, see e.g. \cite[Theorem 13]{Zuniga-LNM-2016}. Furthermore,
there exist positive constants $C_{1}$, $C_{0}$ such that%
\begin{equation}
\frac{C_{1}\sigma}{\left(  \left\vert \xi\right\vert _{p}+\sigma^{\frac
{1}{\alpha}}\right)  ^{1+\alpha}}\leq Z(\xi;\sigma,\alpha)\leq\frac
{C_{0}\sigma}{\left(  \left\vert \xi\right\vert _{p}+\sigma^{\frac{1}{\alpha}%
}\right)  ^{1+\alpha}}, \label{Bound_Heat_Ker}%
\end{equation}
for $\sigma,\alpha>0$, and $\xi\in\mathbb{Q}_{p}$. The upper bound was
established in \cite[Lemma 4.1]{Kochubei}, see also \cite[Theorem
32]{Zuniga-LNM-2016}. The lower bound was established in \cite[Theorem
5.17]{Bendikov et al}.\ In particular, $Z(\cdot;\sigma,\alpha)\in L^{1}$.

Now,%
\begin{align*}
\mathcal{F}_{x\rightarrow\xi}(Q_{0}\left(  \left\vert x\right\vert _{p}%
;\sigma,\alpha\right)  )  &  =\mathcal{NF}_{x\rightarrow\xi}\left(
\exp(-\sigma\left\vert x\right\vert _{p}^{\alpha})\text{\ }\Omega\left(
\left\vert x\right\vert _{p}\right)  \right)  =\mathcal{N}Z(\xi;\sigma
,\alpha)\ast\Omega\left(  \left\vert \xi\right\vert _{p}\right) \\
&  =\mathcal{N}%
{\displaystyle\int\limits_{\mathbb{Q}_{p}}}
Z(\xi-y;\sigma,\alpha)\Omega\left(  \left\vert y\right\vert _{p}\right)  dy,
\end{align*}
and by using the lower bound in (\ref{Bound_Heat_Ker}), and the ultrametric
property of $\left\vert \cdot\right\vert $, and assuming that $\left\vert
\xi\right\vert _{p}>1$,
\begin{multline*}
\mathcal{F}_{x\rightarrow\xi}(Q_{0}\left(  \left\vert x\right\vert _{p}%
;\sigma,\alpha\right)  )\geq\mathcal{N}C_{1}\sigma%
{\displaystyle\int\limits_{\mathbb{Q}_{p}}}
\frac{\Omega\left(  \left\vert y\right\vert _{p}\right)  dy}{\left(
\left\vert \xi-y\right\vert _{p}+\sigma^{\frac{1}{\alpha}}\right)  ^{1+\alpha
}}\\
\geq\mathcal{N}C_{1}\sigma%
{\displaystyle\int\limits_{\substack{y\in\mathbb{Z}_{p}\\\left\vert
\xi\right\vert _{p}>\left\vert y\right\vert _{p}}}}
\frac{\Omega\left(  \left\vert y\right\vert _{p}\right)  dy}{\left(
\left\vert \xi-y\right\vert _{p}+\sigma^{\frac{1}{\alpha}}\right)  ^{1+\alpha
}}=\mathcal{N}C_{1}\sigma%
{\displaystyle\int\limits_{\substack{y\in\mathbb{Z}_{p}\\\left\vert
\xi\right\vert _{p}>\left\vert y\right\vert _{p}}}}
\frac{\Omega\left(  \left\vert y\right\vert _{p}\right)  dy}{\left(
\left\vert \xi\right\vert _{p}+\sigma^{\frac{1}{\alpha}}\right)  ^{1+\alpha}%
}\\
=\frac{\mathcal{N}C_{1}\sigma}{\left(  \left\vert \xi\right\vert _{p}%
+\sigma^{\frac{1}{\alpha}}\right)  ^{1+\alpha}}%
{\displaystyle\int\limits_{y\in\mathbb{Z}_{p}}}
dy=\frac{\mathcal{N}C_{1}\sigma}{\left(  \left\vert \xi\right\vert _{p}%
+\sigma^{\frac{1}{\alpha}}\right)  ^{1+\alpha}}\text{ for }\left\vert
\xi\right\vert _{p}>1\text{.}%
\end{multline*}
Then, we have the following result:

\begin{lemma}
\label{Lemma_F}Take $\sigma,\alpha>0$ as before. Then%
\[
\mathcal{F}_{x\rightarrow\xi}(Q_{0}\left(  \left\vert x\right\vert _{p}%
;\sigma,\alpha\right)  )\geq\frac{\mathcal{N}C_{1}\sigma}{\left(  \left\vert
\xi\right\vert _{p}+\sigma^{\frac{1}{\alpha}}\right)  ^{1+\alpha}}\text{ for
}\left\vert \xi\right\vert _{p}>1\text{.}%
\]

\end{lemma}

On the other hand, since $\left(  \mathbb{Q}_{p},\left\vert \cdot\right\vert
_{p}\right)  $ is a Polish space, a complete, separable metric space, every
probability measure is tight, see e.g. \cite[Proposition 1.3.24]{Kondratiev et
al}, which implies that given $\epsilon>0$, there exists a compact subset
$K_{\epsilon}\subset\mathbb{Z}_{p}$ such that $\int_{K_{\epsilon}}Q_{0}\left(
\left\vert x\right\vert _{p}\right)  dx>1-\epsilon$. Now since $K_{\epsilon}$
is bounded, there exists a non-negative integer $M$ such that $K_{\epsilon
}\subset p^{M}\mathbb{Z}_{p}$, and consequently,%
\begin{equation}%
{\displaystyle\int\limits_{p^{M}\mathbb{Z}_{p}}}
Q_{0}\left(  \left\vert x\right\vert _{p}\right)  dx\geq%
{\displaystyle\int\limits_{K_{\epsilon}}}
Q_{0}\left(  \left\vert x\right\vert _{p}\right)  dx>1-\epsilon.
\label{Fixing_M}%
\end{equation}
By choosing $\epsilon$ so that $1-\epsilon\in\left(  \frac{1}{2},1\right)  $,
the Hypothesis A is satisfied. Notice that the integer $M$ depends on
$\epsilon,\sigma,\alpha$.

Now we proceed to analyze Hypothesis B. By using Lemma \ref{Lemma_F},%
\begin{equation}
\widehat{Q_{0}}\left(  p^{1-r_{0}}\right)  f(I_{0})\geq\frac{\mathcal{N}%
C_{1}\sigma}{\left(  p^{1-r_{0}}+\sigma^{\frac{1}{\alpha}}\right)  ^{1+\alpha
}}\text{,} \label{Eq_INequality_1}%
\end{equation}
since $r_{0}<0$, $\xi=p^{r_{0}-1}$satisfies $\left\vert \xi\right\vert _{p}%
>1$.$\ $Now, we take $\sigma\geq p^{\frac{1-r_{0}}{\alpha}}$, from
(\ref{Eq_INequality_1}) we have%
\[
\frac{\mathcal{N}C_{1}\sigma}{\left(  p^{1-r_{0}}+\sigma^{\frac{1}{\alpha}%
}\right)  ^{1+\alpha}}\geq\frac{\mathcal{N}C_{1}\sigma}{\left(  2\sigma
^{\frac{1}{\alpha}}\right)  ^{1+\alpha}}=\frac{\mathcal{N}C_{1}}{2^{1+\alpha
}\sigma^{\frac{1}{\alpha}}}.
\]
Finally, the Hypothesis B is satisfied if by taking $\frac{\mathcal{N}C_{1}%
}{2^{1+\alpha}\sigma^{\frac{1}{\alpha}}}>\mu_{\max}$, i.e. if
\begin{equation}
\sigma<\sigma_{\max}=:\left[  \frac{\mathcal{N}C_{1}}{2^{1+\alpha}\mu_{\max}%
}\right]  ^{\alpha}. \label{Fixinf_sigma}%
\end{equation}
Given $\alpha>0$, we pick $\sigma>0$ satisfying (\ref{Fixinf_sigma}), i.e.
Hypothesis B is satisfied. Now for $\alpha,\sigma$ fixed, we pick $\epsilon$
so that $1-\epsilon\in\left(  \frac{1}{2},1\right)  $, then there exists an
integer $M$ such that (\ref{Fixing_M}) holds true, i.e. Hypothesis A is
satisfied. Then we have the following result.

\begin{theorem}
\label{Theorem_E}Assume that $X(x,0)\in\mathcal{S}_{0}$, and that the mutation
measure is as in (\ref{Mutation_Measure}), with $\sigma\in\left(
0,\sigma_{\max}\right)  $, $\alpha\in\left(  0,\infty\right)  $, $\int
_{p^{M}\mathbb{Z}_{p}}Q_{0}\left(  \left\vert x\right\vert _{p}\right)
dx\in\left(  \frac{1}{2},1\right)  $, $M=M(\sigma,\alpha)$, then the Cauchy
problem (\ref{eq6}) admits a quasispecies solution.
\end{theorem}


\begin{thebibliography}{99}                                                                                               %
\bigskip

\bibitem {Alberio et al}S. Albeverio, A. Yu. Khrennikov, V. M. Shelkovich,
T\textit{heory of }$p$\textit{-adic distributions: linear and nonlinear
models}, London Mathematical Society Lecture Note Series, 370 (Cambridge
University Press, 2010).

\bibitem {Albeverio-Kozyrev}Albeverio, S., Kozyrev, S.V. Multidimensional
basis of $p$-adic wavelets and representation theory. \textit{p-Adic Num
Ultrametric Anal, Appl.} \textbf{1}, 181--189 (2009).

\bibitem {Av-Zhu}V. A. Avetisov,Yu. N. Zhuravlev, An evolutionary
interpretation of a $p$-adic equation of ultrametric diffusion. \textit{Dokl.
Math.} \textbf{75}, no. 3, 453--455 (2007).

\bibitem {Av-Zhu-2}V. A. Avetisov, Yu. N. Zhuravlev, Hierarchical Scale-Free
Representation of Biological Realm---Its Origin and Evolution in
\textit{Biosphere Origin and Evolution}, eds. N. Dobretsov, N. Kolchanov, A.
Rozanov, G. Zavarzin, (Springer, 2008), pp 69-88, 2008.

\bibitem {Bendikov et al}A. D. Bendikov, A. A. Grigor'yan, K. Pitt\`{e}, V.
V\"{e}ss, Isotropic Markov semigroups on ultra-metric spaces. \textit{Russian
Math. Surveys }\textbf{69}, no. 4, 589--680 (2014).

\bibitem {Bikulov}A. Kh. Bikulov, A. P. Zubarev, Complete systems of
eigenfunctions of the Vladimirov operator in $L^{2}(B_{r})$ and $L^{2}%
(\mathbb{Q}_{p})$. \textit{J. Math. Sci. (N.Y.)} \textbf{237}, no. 3, 362--374 (2019).

\bibitem {Dra-Kh-K-V}B. Dragovich, A. Yu. Khrennikov, S. V. Kozyrev, I. V.
Volovich, On $p$-adic mathematical physics. $p$\textit{-Adic Numbers
Ultrametric Anal. Appl.} \textbf{1}, no. 1, 1--17 (2009).

\bibitem {Dragovich}B. Dragovich and A. Yu. Dragovich, A $p$-adic model of DNA
sequence and genetic code. $p$\textit{-Adic Numbers Ultrametric Anal. Appl.}
\textbf{1}, no. 1, 34--41 (2009).

\bibitem {Eigen1971}M. Eigen, Selforganization of matter and the evolution of
biological macromolecules. \textit{Naturwissenschaften} \textbf{58}, no. 10,
465--523 (1971).

\bibitem {Eigen et al}Manfred Eigen, John McCaskill, Peter Schuster, Molecular
quasi-species. \textit{J. Phys. Chem.} \textbf{92}, no. 24, 6881--6891 (1988).

\bibitem {Halmos}Paul R. Halmos, \textit{Measure Theory} (D. Van Nostrand
Company, 1950).

\bibitem {Horn and Jhonson}Roger A. Horn, Charles R. Johnson, \textit{Matrix
Analysis}. Cambridge University Press. Second Edition (2013)

\bibitem {Igusa}Igusa J.-I., A\textit{n introduction to the theory of local
zeta functions}, in AMS/IP Studies in Advanced Mathematics, 14, American
Mathematical Society, Providence, RI; International Press, Cambridge, MA, 2000

\bibitem {James et al}K. D. James, A. D. Ellington,The Fidelity of
template-directed oligonucleotide ligation and the inevitability of polymerase
function. \textit{Orig. Life Evol. Biosph.} \textbf{29}, 375-390 (1999).

\bibitem {Kochubei}Anatoly N. Kochubei, \textit{Pseudo-differential equations
and stochastics over non-Archimedean fields} (Marcel Dekker, 2001).

\bibitem {Kosyak et al}A. Y. Khrennikov, A. V. Kosyak, V. M. Shelkovich,
Wavelet analysis on adeles and pseudo-differential operators. \textit{J.
Fourier Anal. Appl.} \textbf{18}, no. 6, 1215--1264 (2012).

\bibitem {KKGentic}A. Yu. Khrennikov and S. V. Kozyrev, Genetic code on the
dyadic plane. \textit{Physica A: Stat. Mech. Appl.} \textbf{381}, 265--272
\textbf{ }(2007).

\bibitem {KKZuniga}Andrei Khrennikov, Sergei Kozyrev, W. A. Z{\'{u}}%
{\~{n}}iga-Galindo, \textit{Ultrametric Equations and its Applications},
Encyclopedia of Mathematics and its Applications (168) (Cambridge University
Press, 2018).

\bibitem {Koblitz}Neal Koblitz,\textit{ }$p$\textit{-adic Numbers, }%
$p$\textit{-adic Analysis, and Zeta-Functions}, Graduate Texts in Mathematics
No. 58 (Springer-Verlag, 1984).

\bibitem {Kondratiev et al}Sergio Albeverio, Yuri Kondratiev, Yuri Kozitsky,
Michael R\"{o}ckner, T\textit{he statistical mechanics of quantum lattice
systems. A path integral approach}. EMS Tracts in Mathematics, 8. European
Mathematical Society (EMS), Z\"{u}rich, 2009.

\bibitem {Kozyrev  SV}S. V. Kozyrev, Methods and Applications of Ultrametric
and $p$-Adic Analysis: From Wavelet Theory to Biophysics. \textit{Proc.
Steklov Inst. Math.} \textbf{274}, Suppl. 1, 1-84 (2011).

\bibitem {L-N}Pietro Li\`{o} and Nick Goldman, Models of Molecular Evolution
and Phylogeny. \textit{Genome Res.} \textbf{8}, 1233-1244 (1998).

\bibitem {Nowak}Martin A. Nowak, \textit{Evolutionary dynamics. Exploring the
equations of life} (Harvard University Press, 2006).

\bibitem {Perko}Lawrence Perko, \ \textit{Differential equations and dynamical
systems}. 3rd ed. Texts in Applied Mathematics. 7. New York, Springer (2001).

\bibitem {Poole etal}A. Poole, D. Jeffares, D. Penny, Early evolution:
Prokaryotes, the new kids on the block. \textit{BioEssays} \textbf{21},
880--889 (1999).

\bibitem {Sch}I. Scheuring, Avoiding Catch-22 of early evolution by stepwise
increase in copying fidelity. \textit{Selection} \textbf{1}, 13-23 (2000).

\bibitem {Maynad Smith}J. Maynard Smith, Models of Evolution. \textit{Proc. R.
Soc. Lond.,} Series B, Biological Sciences, \textbf{219}, no. 1216, 315--325
\textbf{ }(1983).

\bibitem {Saakian1}David B Saakian and Chin-Kun Hu, Exact solution of the
Eigen model with general fitness functions and degradation rates,
\textit{Proc. Natl. Acad. Sci. USA} \textbf{103}, no 13, 4935-4939 \textbf{ }(2006).

\bibitem {SchusterPeter}Peter Schuster, The Mathematics of Darwin's Theory of
Evolution: 1859 and 150 Years Later in \textit{The Mathematics of Darwin's
Legacy}, Mathematics and Biosciences in Interaction eds. F. Chalub, J.
Rodrigues (Springer, 2011), pp. 27-66.

\bibitem {Szat-PTRSB}E. Szathm\'{a}ry,The origin of replicators and
reproducers. \textit{Philosophical Transactions of the Royal Society B:
Biological Sciences} \ \textbf{361}, 1761--1776 (2006).

\bibitem {Szat-tree}E. Szathm\'{a}ry, The integration of earliest genetic
information. \textit{Trends Ecol. Evol. }\textbf{4}, 200-204 (1989).

\bibitem {Taibleson}M. H. Taibleson, \textit{Fourier analysis on local fields}
(Princeton University Press, 1975).

\bibitem {Tannembaum et al}Emmanuel Tannenbaum, Eugene I. Shakhnovich,
Semiconservative replication, genetic repair, and many-gened genomes:
Extending the quasispecies paradigm to living systems. \textit{Phys. Life
Rev.} \textbf{2}, 290-317 (2005).

\bibitem {V-V-Z}V. S. Vladimirov, I. V. Volovich, E. I. Zelenov,
$p$\textit{-adic analysis and mathematical physics} (World Scientific, 1994).

\bibitem {zuniga-Nonlieal}W. A. Z{\'{u}}{\~{n}}iga-Galindo, Non-Archimedean
Reaction-Ultradiffusion Equations and Complex Hierarchic Systems.
\textit{Nonlinearity} \textbf{31}, no. 6, 2590--2616 (2018).

\bibitem {Zuniga-LNM-2016}W. A. Z\'{u}\~{n}iga-Galindo,
\textit{Pseudodifferential equations over non-Archimedean spaces}, Lectures
Notes in Mathematics 2174 (Springer, Cham, 2016).

\bibitem {Zuniga-JPhA}W. A. Z\'{u}\~{n}iga-Galindo, Non-Archimedean replicator
dynamics and Eigen's paradox. \textit{J. Phys. A: Math. Theor.} \textbf{51},
505601 2018.
\end{thebibliography}
\end{document}